%% file: 2026-LewisWeights.tex
\documentclass[11pt]{article}
\input{packages}
\input{formatting_macros}

\usepackage{thm-restate}

\input{macros}

\title{Computing Lewis Weights to High Precision by Fixed-Point Iteration}
\author{Swati Padmanabhan\thanks{Department of Industrial and Systems Engineering, University of Minnesota, Twin Cities. pswt@umn.edu}}
\date{}

\begin{document}

\maketitle

\begin{abstract}
\looseness=-1The $\ell_p$-Lewis weights of a matrix are measures of row importance that generalize leverage scores to $\ell_p$ geometry and are widely used in matrix sampling and regression. They are defined by a fixed-point equation, which suggests a particularly simple algorithm to compute them: repeatedly apply the defining map. This raises the following natural question: \emph{Can this direct iteration compute the weights to high precision?}  For $p<4$, Cohen and Peng~\cite{cohen2015lp} showed that iterating an equivalent rearrangement of this equation succeeds; however, for $p\geq4$, this analysis breaks down, and prior high-precision methods instead use sophisticated optimization-based approaches. We show that the direct Lewis fixed-point iteration, appearing in Lee~\cite{lee2016faster}, does in fact compute Lewis weights to high precision for every $p>2$. 

For a matrix $\bA\in\mathbb R^{m\times n}$ partitioned into row blocks $\bA_{[1]},\ldots,\bA_{[k]}$, we compute, for $p>2$, coordinatewise $\varepsilon$-approximate $\ell_p$ block Lewis weights in $O\left(p\log\frac{p\sqrt{ \sum_{i=1}^k\rank(\bA_{[i]}) }}{\varepsilon}\right)$  rounds of exact leverage-score-vector computations. For row-wise Lewis weights, this becomes $O\left(p\log\frac{p\sqrt m}{\varepsilon}\right)$, improving the $O(p^2\log(m/\varepsilon))$ bound of Gribling, Sidford, and Zhang~\cite{pmlr-v336-gribling26a} for $p\geq4$. Our guarantees also extend to   approximate leverage-score computations. 

\looseness=-1Our main insight is that each update contracts the KL divergence to the true weights by a factor of $1-2/p$. Our proof of this contraction is extremely simple, combining Jensen's, matrix Cauchy-Schwarz, and H\"older's inequalities. We provide an additional proof  through volume sampling and entropic independence. Synthetic experiments closely match our predicted rates; real-data experiments illustrate the information-concentration tradeoff of finite-$p$ block Lewis designs.
\end{abstract}

\clearpage
\pagenumbering{arabic}

\section{Introduction}
\label{sec:intro}
\looseness=-1A central question in data science, when working with a large matrix $\bA$, is how to determine the importance of each row, for example when sampling and rescaling rows to obtain a smaller matrix that approximately preserves $\|\bA\bx\|_p$ for every $\bx$. For $p=2$, \emph{leverage scores} provide a natural measure of row importance for this purpose. The $\ell_p$ \emph{Lewis weights}, which we study in this paper, generalize leverage scores to any $p>0$ and are widely used in matrix sampling, $\ell_p$ regression, and optimization.

\looseness=-1Lewis weights are defined by a fixed-point equation~\Cref{eq:Lewis-fixed-point}, suggesting an especially simple approach to computing them: directly iterate the right-hand side of this equation. This raises the following natural question: \emph{Can this iteration compute Lewis weights to high precision?} Cohen and Peng~\cite{cohen2015lp} showed that iterating an equivalent rearrangement of \Cref{eq:Lewis-fixed-point} succeeds for $p<4$. For $p\geq4$, however, their analysis breaks down, and subsequent high-precision algorithms instead rely on sophisticated optimization-based approaches  (cf. \Cref{sec:intro-approachOverview}). The main message of this paper is that the direct fixed-point iteration associated with \Cref{eq:Lewis-fixed-point},
appearing in~\cite{lee2016faster}, does in fact compute
Lewis weights to high precision for $p>2$, thereby answering
this question affirmatively in this regime as well.

\begin{definition}[Lewis weights~\cite{lewis1978finite,cohen2015lp}]\label{def:LW}
Given a full-column-rank matrix $\bA\in\R^{m\times n}$ with no zero rows, and a scalar $p>0$, the $\ell_p$-Lewis weights of $\bA$ are the entries of the unique vector $\wbar\in\Rmspos$ satisfying the fixed-point equation 
\[\wbar_i=\levsc{i}\left(\Wbar^{1/2-1/p}\bA\right)
\text{ for all }i\in[m], \numberthis\label{eq:Lewis-fixed-point}
\]
where  $\Wbar$ is the diagonal matrix with vector $\bwbar$ along its diagonal,  $\levsc{i}(\bB):=\bb_i^\top(\bB^\top\bB)^{-1}\bb_i$ denotes the $i^\mathrm{th}$ leverage score of a full-column-rank matrix $\bB$ (here $\bb_i^\top$ is the $i^\mathrm{th}$ row of $\bB$). Throughout this paper, we refer to \Cref{eq:Lewis-fixed-point} as the ``Lewis fixed-point equation''. 
\end{definition} \looseness=-1The existence of these weights was shown by Lewis~\cite{lewis1978finite},
while their uniqueness was subsequently established by Schechtman
and Zvavitch~\cite{schechtman2001embedding} and, via a different
argument, by Cohen and Peng~\cite{cohen2015lp}. We provide a self-contained explanation of leverage scores and Lewis weights in \Cref{app:row-splitting}; this is not necessary for the remainder of this paper, but would be useful to make sense of \Cref{eq:Lewis-fixed-point}.  

\looseness=-1As explained in \Cref{{sec:rowsplit}}, sampling (and appropriately scaling) the rows of a tall, full-column-rank matrix proportional to their Lewis weights  can yield a smaller matrix that approximately preserves the $\ell_p$ geometry of the original.   This approach has yielded nearly optimal bounds on the sample
complexity of $\ell_p$ subspace embeddings, even in online and
streaming settings, and on the query complexity of active
$\ell_p$ regression~\cite{cohen2015lp,musco2022active,woodruff2025online}.  Strikingly, Lewis-weight reweighting underlies Lee and Sidford's breakthrough $\widetilde{O}(\sqrt{\operatorname{rank}(\bA)})$-iteration linear-programming algorithm and their nearly optimal self-concordant barrier for polytopes~\cite{lee2019solving}, which then enabled $\widetilde{O}(n^2)$-step Dikin walks for sampling from arbitrary $n$-dimensional polytopes when initialized from a warm start~\cite{laddha2020strong, chen2026two} (cf. \Cref{sec:apps} for more applications). Together, these examples motivate the design of efficient algorithms
for computing Lewis weights. Our focus is on high-precision
computation, specifically on achieving logarithmic dependence on
$1/\varepsilon$ in the number of leverage-score computations,
where $\varepsilon$ is the desired relative error.

The fixed-point equation  in \Cref{def:LW} extends beyond the row-wise setting to a natural block analogue, in which a single weight is assigned to each block according to its total leverage score after blockwise rescaling. In this paper, we restrict ourselves to $p > 2$, which is therefore also the regime for which we state \Cref{def:blockLW}, the block generalization of \Cref{def:LW}.
\begin{definition}[Block Lewis weights~\cite{jambulapati2023sparsifying}]\label{def:blockLW} 
Given a full-column-rank matrix $\bA\in\R^{m \times n}$,  a partition of its rows into $k$ nonzero blocks $\bA_{[1]},\bA_{[2]},\dotsc,\bA_{[k]}$ given as part of the input, with $\bA_{[i]}\in\R^{m_i\times n}$, and a scalar $p >2$, the block $\ell_p$-Lewis weights of $\bA$ are the entries of the  unique   vector $\wbar\in\Rkpos$ satisfying \[\wbar_i = \blevsc{i}\left(\Wbar^{1/2-1/p}\bA\right) \text{ for all } i \in [k],\numberthis\label{eq:blk-LW-fp-eq}\] where $\bWbar:=\Diag(\wbar_1\bI_{m_1},\dotsc,\wbar_k\bI_{m_k})$, and
$\blevsc{i}(\bB):=\sum_{j \in I_i}\levsc{j}(\bB)$ denotes the total leverage score of the $i^\mathrm{th}$ block of $\bB$ (with row-index set $I_i\subseteq [m]$).
\end{definition}  
When each block in \Cref{def:blockLW} comprises just a row, \Cref{def:blockLW} reduces to the (row-wise) Lewis weights from \Cref{def:LW}. Block Lewis weights were introduced in~\cite{jambulapati2023sparsifying} for sparsifying sums of norms, and their construction also establishes existence; uniqueness for $p>2$ follows from \Cref{lem:kl-contraction}. They have since been further developed for matrix block-norm approximation and distributionally robust regression~\cite{manoj2025change,manoj2026distributionally}. We use the block formulation as a common setting in which to study the fixed-point dynamics of both row-wise and block Lewis weights.

\subsection{Our contribution}
\label{sec:intro-contribution}
In this subsection, we give an overview of our theoretical and experimental results. 

\paragraph{(1) Theoretical results.}  
Our main result is a guarantee for computing $\varepsilon$-approximate block Lewis weights for $p>2$. Our algorithm (displayed in \Cref{alg:main}) iterates the fixed-point equation \Cref{eq:blk-LW-fp-eq}; the row-wise version of this iteration appears in~\cite[Algorithm~6]{lee2016faster}, where iterates are averaged to yield one-sided approximate Lewis weights.  We show that, with appropriate initialization and final output conversion, this iteration yields high-precision block Lewis weights:

\begin{theorem}[Block Lewis weights]\label{thm:blkLW}
Let $\bA\in\R^{m\times n}$ be a full-column-rank matrix partitioned into $k$ nonzero blocks $\bAblk{1}, \bAblk{2}, \ldots,\bAblk{k}$ with $\bAblk{i}\in\R^{m_i\times n}$ for all $i\in [k]$. Assume  $p>2$, and let $\bwbar\in\Rkpos$ denote the $\ell_p$-block-Lewis weights of $\bA$. Given   $\eps\in(0,1)$, the block-wise Lewis fixed-point iteration (\Cref{alg:main}) returns $\bwhat\in\Rkpos$ satisfying $(1-\eps)\wbar_i \le \what_i \le (1+\eps)\wbar_i$ for all $i\in[k]$ using \[O\left(
    p\log\left(\frac{p \sqrt{\sum_{i=1}^k \rank\left(\bAblk{i}\right)}}{\eps}\right)\right) \text{ exact leverage-score-vector computations}\] of a matrix $\bD\bA$, where  $\bD = \Diag\left(d_1\bI_{m_1}, d_2\bI_{m_2}, \dotsc, d_k\bI_{m_k}\right)$ for some positive $d_1, d_2, \dotsc, d_k$. 
\end{theorem}

\noindent Specializing each block to a single row gives the following guarantee. 
\begin{corollary}
\label{cor:ordinary-main}
Let  $\bA\in\R^{m\times n}$ be a full-column-rank matrix with no zero rows, let $p>2$, and let $\bwbar\in\Rmspos$ denote its $\ell_p$-Lewis weights. Given $\eps\in(0,1)$, the Lewis fixed-point iteration (\Cref{alg:main} specialized to  $k=m$) returns $\bwhat\in\Rmspos$ satisfying $(1-\eps)\wbar_i \le \what_i \le (1+\eps)\wbar_i$ for all $i\in[m]$ 
using \[O\left(
    p\log\left(\frac{p\sqrt{m}}{\eps}\right)
  \right) \text{ exact leverage-score-vector computations}\]  of $\bD\bA\in\R^{m\times n}$ where $\bD = \Diag\left(d_1, d_2, \dotsc, d_m\right)$ for some positive $d_1, d_2, \dotsc, d_m$. 
\end{corollary} \looseness=-1For $p\geq4$, \Cref{thm:blkLW,cor:ordinary-main} give nearly dimension-free high-precision guarantees for the direct Lewis fixed-point iteration. To our knowledge, no such guarantee was previously known. 
Concurrent independent work of Chen also obtains the guarantee in \Cref{cor:ordinary-main} (cf. \Cref{sec:intro-approachOverview}).  While our above results require \emph{exact} leverage scores, we show in \Cref{thm:approx-levscores-blkLW} that the Lewis fixed-point iteration is also robust to the use of \emph{approximate} leverage scores. 
\begin{theorem}[Block Lewis weights via approximate leverage scores]\label{thm:approx-levscores-blkLW} Consider the setup of \Cref{thm:blkLW}. The approximate-leverage-score variant of \Cref{alg:main} described in \Cref{sec:approxLevScores} returns  $\bwhat\in\R^k_{>0}$ satisfying $(1-\varepsilon)\wbar_i\leq \what_i\leq (1+\varepsilon)\wbar_i$ for all $i\in [k]$ using \[O\left(
    p\log\left(\frac{p \sqrt{\sum_{i=1}^k \rank\left(\bAblk{i}\right)}}{\eps}\right)\right) \text{ approximate leverage-score-vector computations}.\] Each such call computes, with coordinatewise relative error $O\left(\frac{\varepsilon}{p^2\sqrt{n}}\right)$,  leverage scores of a matrix $\bD\bA$, where $\bD = \Diag\left(d_1\bI_{m_1}, d_2\bI_{m_2}, \dotsc, d_k\bI_{m_k}\right)$ for some positive $d_1, d_2, \dotsc, d_k$. 
\end{theorem}

\begin{remark}  \looseness=-1We now further contextualize the above results. 

\begin{enumerate}[
    label=(\roman*),
    leftmargin=*,
    topsep=4pt,
    itemsep=2.5pt,
    parsep=0pt
]
\item\looseness=-1 Our analysis is in exact real arithmetic, and we make no
claims about bit complexity or floating-point stability.
With exact leverage scores, the total arithmetic complexity   is $O\left(p mn^{\omega-1}\log\left(\frac{p\sqrt{m}}{\varepsilon}\right)\right)$, where $\omega$ is the matrix-multiplication exponent~\cite{huang1998fast, alman2025more}.
\item For $p\geq 4$, when computing using exact leverage scores, \Cref{cor:ordinary-main} improves upon the ${O}\left(p^2\log\left(\frac{m}{\varepsilon}\right)\right)$ iteration
complexity of \cite{pmlr-v336-gribling26a}.  
\item For $p\geq 4$, when computing using approximate leverage scores, 
\Cref{thm:approx-levscores-blkLW}, specialized to the row-wise
setting, improves upon the leverage-score-vector oracle complexity of~\cite{pmlr-v336-gribling26a},
which computes a two-sided $\varepsilon$-approximation to the Lewis
fixed-point equation using
$O\left(\frac{pn\log m}{\varepsilon}\right)$
leverage-score-vector queries, each at relative accuracy
$O\left(\frac{\varepsilon}{pn}\right)$.
\item  Our results have implications for the estimation of $\ell_p$ sensitivities, regularized D-optimal design, and the evaluation of the Lewis weights barrier (cf. \Cref{sec:apps}).
\end{enumerate}
\end{remark} 
\paragraph{Proof technique.} The key observation behind all above guarantees is that the direct
Lewis update satisfies a \emph{Kullback-Leibler (KL) contraction}.  Writing $\calT$ for one such update, we prove in \Cref{lem:kl-contraction} that 
\[
    \kldiv{\bwbar}{\calT(\bw)}\leq \left(1-\frac{2}{p}\right)\kldiv{\bwbar}{\bw}.
\]This makes the direct Lewis update contractive,  in \emph{KL divergence
measured from the true weights}, for $p>2$; this includes the regime $p\geq4$  where the update based on an equivalent reformulation of the Lewis fixed-point equation analyzed by Cohen and Peng~\cite{cohen2015lp} is no longer guaranteed to contract  \emph{coordinatewise multiplicative error} (see \Cref{sec:intro-approachOverview}). It is the above inequality that drives our $\widetilde{O}(p)$ rate in both \Cref{thm:blkLW,thm:approx-levscores-blkLW}. Our first proof of this contraction (cf. \Cref{sec:kl-decay}) is extremely simple, combining Jensen's, matrix Cauchy-Schwarz, and H\"older's inequalities.

\looseness=-1In addition to this direct proof in \Cref{sec:kl-decay}, we  present in \Cref{sec:global-intuition} an alternate proof of the above KL contraction for the exact Lewis update by interpreting it via volume sampling and invoking the property of  entropic independence~\cite{derezinski2021determinantal,anari2022entropic}. We believe this  alternate perspective  makes the origin of the contraction factor $1-\frac{2}{p}$ more intuitive.

\looseness=-1While the above contraction yields the stated \emph{global} iteration bounds, it does not characterize the asymptotic convergence rate near the fixed point $\bwbar$. We therefore also study the \emph{local} dynamics of the row-wise iteration by linearizing the Lewis update at $\bwbar$ (\Cref{{sec:local-dyn}}).
The spectral radius of the resulting Jacobian predicts the asymptotic local contraction rate, which we  corroborate experimentally.

\paragraph{(2) Experimental results.} We complement our theoretical results with synthetic and real-data experiments. Our synthetic experiments first test the local fixed-point dynamics mentioned above.
Near the fixed point, the observed contraction rates closely agree with the rates predicted by our local analysis (\Cref{fig:local-dyn}).
Beyond these local dynamics, we study the total iteration count. For row-wise Lewis weights on i.i.d.\ Gaussian matrices, the number of iterations grows approximately linearly with $p$ (\Cref{fig:A1}). This qualitative behaviour persists on a structured matrix family with nonuniform leverage scores (\Cref{fig:structured-ordinary}). We also test the fixed-point iteration in floating-point arithmetic on ill-conditioned matrices with identical exact Lewis weight trajectories; the observed iteration counts remain essentially unchanged over a broad range of matrix condition numbers, until finite precision imposes an accuracy floor (\Cref{fig:a4-map-error,fig:a4-residual-floor}). Finally, we observe an analogous dependence on $p$ when testing the  \emph{block} Lewis weight iteration count  as well (\Cref{fig:block-scaling}).

Having tested the convergence behaviour of the Lewis  fixed-point iteration,  we next use it to study the \emph{properties} of the block Lewis designs it computes.   In particular, motivated by the connection between Lewis weights and regularized D-optimal design (cf. \Cref{sec:apps}), we study these designs  on the Intel Berkeley Research Lab sensor dataset~\cite{guestrin2004distributed} and meteorological data from the NOAA National Data Buoy Center~\cite{ndbc1971}. 
 As $p$ increases, the  designs generated in our experiments move from diffuse, near-uniform allocations toward the more concentrated D-optimal design, while recovering much of the available D-optimal information gain (\Cref{fig:design-geometry}).

This suggests that finite, but large, values of $p$ retain most of the information captured by D-optimal design while distributing weights more broadly, and hence are more  robust to the loss of individual blocks. To explicitly test this possibility, we evaluate our generated designs under different patterns of block unavailability. When few blocks are unavailable, larger values of $p$ preserve more information; as the number of unavailable blocks increases, the values of $p$ preserving the most information  systematically decrease (\Cref{fig:block-failure}). We again stress that this particular set of experiments is not intended to test our algorithm per se; rather, it studies properties of the finite-$p$ block Lewis designs. Indeed, these experiments could in principle be conducted using any method that computes Lewis weights to high precision; however,  the \emph{simplicity} of our algorithm makes it straightforward to carry out this computation for a wide range of values of $p$. 

Overall, our experiments illustrate both  the fixed-point dynamics  as well as the   information-concentration  tradeoff exhibited by the finite-$p$ block Lewis designs our algorithm computes.

\subsection{Prior work}\label{sec:intro-approachOverview}
\looseness=-1Prior work on computing Lewis weights to high precision exhibits a sharp difference between the regimes  $p<4$ and  $p\geq 4$.

\paragraph{The regime $p<4$.}
For $p<4$, Cohen and Peng~\cite{cohen2015lp} analyze an equivalent rearrangement of \Cref{eq:Lewis-fixed-point}. 
In particular, one update of their algorithm takes the form
\[w_i^+ := \left(\ba_i^\top
\left(\bA^\top\bW^{1-2/p}\bA\right)^{-1}
\ba_i\right)^{p/2}.\]
They show that if the current weights are within a multiplicative factor $c$ of the true Lewis weights, then one such update improves this factor to $c^{|p/2-1|}$. 
Thus, the coordinatewise multiplicative error contracts when $|\frac{p}{2}-1|<1$, or equivalently when $p<4$; at $p=4$ the factor reaches $1$, and for $p>4$ it exceeds $1$. 
Consequently, for $p\geq4$, this analysis does not prove high-precision convergence.

\paragraph{The regime $p\geq4$.} 
For $p\geq4$, subsequent algorithms instead turned to optimization formulations of Lewis weights. 
~\cite{cohen2015lp} use the determinant-maximization formulation together with cutting-plane methods; Lee~\cite{lee2016faster} develops a mirror-descent-based approach; and Lee and Sidford~\cite{lee2019solving} optimize a volumetric potential, using a homotopy argument to obtain a sufficiently good initialization. 
These methods establish efficient computation in this regime, but do not give a nearly dimension-free number of leverage-score computations with logarithmic dependence on $\frac{1}{\eps}$.

\paragraph{High precision for $p\geq4$.} 
Fazel, Lee, Padmanabhan, and Sidford~\cite{fazel2022computing} gave the first algorithm for $p\geq4$ that computes Lewis weights to high precision using a nearly dimension-free number of leverage-score-vector computations, using
$O\left(p^3\log\left(\frac{mp}{\eps}\right)\right)$
such computations. 
Their algorithm minimizes a regularized log-determinant objective via a quasi-Newton-type update. 
Sufficient progress is guaranteed under a certain rounding condition, which is either restored through a separate rounding procedure or maintained by choosing coordinate-dependent step sizes. In recent work, Gribling, Sidford, and Zhang~\cite{pmlr-v336-gribling26a} noted that the objective in a convex formulation of the Lewis weight problem satisfies a certain local relative smoothness property, which in turn enables a faster algorithm using $O\left(p^2\log\left(\frac{m}{\eps}\right)\right)$ leverage-score-vector computations. 

\paragraph{Block Lewis weights.}\looseness=-1Block Lewis weights were introduced by Jambulapati, Lee, Liu, and Sidford~\cite{jambulapati2023sparsifying} for
sparsifying sums of norms and further developed by Manoj and Ovsiankin~\cite{manoj2025change} for matrix block-norm sparsification and change-of-measure arguments. 
The algorithms in these works compute block Lewis overestimates sufficient for these applications, rather than coordinatewise high-precision approximations to the block Lewis weights, which is our focus in \Cref{thm:blkLW}.

\paragraph{Independent concurrent work.}\looseness=-1After the first version of this paper appeared on arXiv, Li Chen
informed us of independent concurrent work and privately shared
their manuscript. It establishes the row-wise
guarantee in \Cref{cor:ordinary-main}, also via a KL contraction. The proofs of this contraction differ: Chen's uses the Ball-Barthe
inequality, while our first proof is linear algebraic, combining
Jensen's, matrix Cauchy-Schwarz, and H\"older's inequalities
(\Cref{sec:kl-decay}); our second proof uses volume sampling
and entropic independence (\Cref{sec:global-intuition}).
The manuscripts also differ in scope: Chen's focuses on the row-wise
setting, whereas the first arXiv version
of this paper (v1) additionally included the extension to
block Lewis weights (\Cref{thm:blkLW}),
a two-sided comparison with $\ell_p$ sensitivities and an improved
approximation guarantee for maximum sensitivity (\Cref{prop:sens}),
an analysis of local convergence (\Cref{lem:local-dyn}),
and numerical experiments on synthetic and real data
(\Cref{sec:expts}).

\subsection{Applications}\label{sec:apps}

Lewis weights play several conceptually distinct roles in algorithms and optimization, from data-dependent importance scores, to experiment design, to self-concordant barrier construction. Below, we review a small subset of these connections.

\paragraph{Sensitivity estimation.}
Sensitivity sampling is an importance-sampling framework in which data points are sampled according to the extent to which they can influence the objective, and has been used extensively in  coresets~\cite{langberg2010universal, varadarajan2012sensitivity}, clustering~\cite{feldman2011unified}, logistic regression~\cite{munteanu2018coresets}, and $\ell_p$ subspace embeddings~\cite{wy2023ICML}. 
The $\ell_p$ sensitivity of the $i^\mathrm{th}$ row of $\bA$ is defined as
$\sens{i}(\bA) := \sup_{\bx\neq \bzero} \frac{|\ba_i^\top \bx|^p}{\|\bA\bx\|_p^p}$.  
Unlike Lewis weights, computing sensitivities directly generally requires solving a separate maximization problem for each row, and recent work has therefore studied efficient methods for approximating them~\cite{padmanabhan2023computing}. 
 Lewis weights are closely related to sensitivities and can therefore serve as efficiently computable surrogates for them. In particular, for $p>2$, it is known that $\sens{i}(\bA) \leq n^{p/2-1}\wbar_i$ \cite[Fact~3.3]{padmanabhan2023computing}. Via a simple ``witness'' argument similar in spirit to that in
\cite[Fact~A.1]{padmanabhan2023computing}, we additionally observe that
the reverse inequality holds as well:
 
\begin{restatable}{proposition}{sensprop}\label{prop:sens}
Let $\bA\in\R^{m\times n}$ have full column rank and no zero rows, let $p>2$, and let $\bwbar\in\Rmspos$ denote its $\ell_p$-Lewis weights. Then, for every $i\in[m]$, we have  $\wbar_i\leq \sens{i}(\bA)$. Combined with the known bound $\sens{i}(\bA)\leq n^{p/2-1}\wbar_i$, this gives
\[\wbar_i\leq \sens{i}(\bA)\leq n^{p/2-1}\wbar_i.\]
Consequently, if, for some $\varepsilon\in (0, 1)$, $\bwhat$ is an $\eps$-approximation to $\bwbar$, i.e.,
$(1-\eps)\wbar_i\leq\what_i\leq(1+\eps)\wbar_i$ for every $i\in[m]$, then, simultaneously for every $i\in[m]$, we have 
$\frac{\what_i}{1+\eps}
\leq
\sens{i}(\bA)
\leq
\frac{n^{p/2-1}\what_i}{1-\eps}.$
\end{restatable}
 
 The first inequality in \Cref{prop:sens} turns the previously known one-sided Lewis weight surrogate for sensitivities into a two-sided comparison. Combined with \Cref{cor:ordinary-main}, it yields simultaneous bounds on all $m$ sensitivities using only a high-precision Lewis weight computation, with no additional individual sensitivity evaluations. In particular, defining $\maxsens(\bA):=\max_{i\in[m]}\sens{i}(\bA)$ and taking maxima in \Cref{prop:sens} gives an $n^{p/2-1}(1+O(\eps))$-factor approximation to $\maxsens(\bA)$. The best prior general-purpose guarantee we are aware of for estimating maximum $\ell_p$ sensitivity without explicitly computing every individual sensitivity is the $O(n^{p/2})$ factor of~\cite{padmanabhan2023computing} for $p>2$. Thus, our guarantee improves the dimension dependence of the approximation factor. 

\paragraph{Regularized D-optimal design.}
D-optimal experimental design is a classical problem in statistics  in which resources are allocated among a collection of experiments so as to maximize the determinant of the resulting information matrix ~\cite{pukelsheim2006optimal,pmlr-v70-allen-zhu17e,pmlr-v99-madan19a,singh2020approximation}. 
It is closely connected to the John ellipsoid problem~\cite{john1948extremum}, which is dual to a relaxed D-optimal-design problem~\cite{Todd2016}. 

For $p>2$, computing the $\ell_p$-Lewis weights is equivalent, after a change of variables, to solving a $\frac{p}{p-2}$-regularized D-optimal-design problem~\cite{cohen2015lp,lee2019solving,fazel2022computing}. 
The regularizer incentivizes spreading the allocation across experiments ---  as $p$ increases, the exponent $\frac{p}{p-2}$ decreases toward $1$, weakening the incentive, and the solution approaches the unregularized D-optimal design. 
This  interpolation between concentrated and uniform allocations motivates our experiments in \Cref{sec:real-data-expts}, where we empirically study how finite-$p$ block Lewis designs vary with $p$.

In light of the above connection, \Cref{cor:ordinary-main} gives a high-precision solver for this regularized design problem using $
O\left(p\log\frac{p\sqrt m}{\delta}\right)$ 
exact leverage-score-vector computations, where $\delta$ denotes the desired coordinatewise accuracy of the corresponding Lewis weights. 

\paragraph{Lewis barriers and Dikin walks.}
Self-concordant barriers are a foundational tool in convex optimization~\cite{nesterov1994interior} and also underlie many algorithms for sampling from polytopes. 
For polytopes, \cite{lee2019solving} introduced a nearly optimal self-concordant barrier based on $\ell_p$-Lewis weights with $p=\Theta(\log m)$; in particular, their barrier is $O(n\log^5 m)$-self-concordant. Lewis weight barriers and Dikin-type walks have since been used in many sampling algorithms~\cite{laddha2020strong,gatmiry2023sampling,jiang2024regularized,kook2024gaussian}. For the~\cite{lee2019solving} barrier, accurately computing the underlying Lewis weights is an important subroutine. Specializing \Cref{cor:ordinary-main} to $p=\Theta(\log m)$ reduces this computation to 
$O\left(\log m\log\frac{\sqrt m\log m}{\delta}\right)$
exact leverage-score-vector computations for coordinatewise relative accuracy $\delta$. 

\subsection{Notation and preliminaries}\label{sec:intro-notationPrelims}

We  boldface non-scalar quantities, with  lowercase letters for vectors and  uppercase letters for matrices. We use $\one$ for the vector of all ones, $\bI$ for the identity matrix, and $\bzero$ for the zero vector or matrix;  sizes are implied by context. For compatible symmetric matrices $\bX,\bY$, we use $\bX\succeq\bY$ to mean that $\bX-\bY$ is positive semidefinite. For a positive integer $s$, we use $[s]$ for the set $\{1,2,\dotsc,s\}$. For vectors, the operations of exponentials, logarithms, and powers are applied coordinatewise. 

Throughout, we assume $p>2$. We use $\bA\in \R^{m\times n}$ to denote our full-column-rank input matrix. 
When discussing block Lewis weights, we assume that a partition
of the rows of $\bA$ into $k$ nonzero blocks is given as part
of the input, and write $\bAblk{i}\in\R^{m_i\times n}$ for
the $i^\mathrm{th}$ block. We use $\bwbar\in\Rmspos$ and $\bwbar\in\Rkpos$ for the row-wise and block-wise weights, respectively, with sizes implied by the context.

We use $\levsc{i}(\bB)$ to denote the $i^\mathrm{th}$ leverage score of a full-column-rank matrix $\bB$. When its rows are partitioned into blocks, we denote the $i^\mathrm{th}$ block leverage score by $\blevsc{i}(\bB):=\Tr\left(\bBblk{i}(\bB^\top\bB)^{-1}\bBblk{i}^\top\right)$. We have a slight notational overloading for diagonal matrices: for a row-wise weight vector (e.g., $\bw\in\R^m$), we use its uppercase boldfaced form ($\bW$ in this example) to denote the diagonal matrix $\bW_{ii}=w_i$; in the block setting, for the block-wise weight vector $\bw\in\R^k$, we instead use $\bW=\Diag\left(w_1\bI_{m_1},w_2\bI_{m_2},\dotsc,w_k\bI_{m_k}\right)$, so that the weight associated with a block is repeated on each of its rows. The intended meaning will be clear from the dimension of the weight vector. We use $w_i^{(t)}$ to denote the $i^\mathrm{th}$ entry of $\bw^{(t)}$, where the superscript $t$ indexes the algorithm's iteration.

Leverage scores (and their block analogues) satisfy classical identities such as in \Cref{fact:block-levsc}, which for completeness we prove in \Cref{sec:app-technical-lemmas}. In particular, the second identity specializes, for graph incidence matrices, to Foster's theorem~\cite{foster1949average}. The formulation of block leverage scores we use previously appeared in~\cite{KLPSS16,xu2016sub} and has been widely used in several forms in, e.g., matrix sparsification and numerical linear algebra.

\begin{restatable}{fact}{blocklevscfact}\label{fact:block-levsc}
Let $\bB:=\begin{bmatrix}\bB_{[1]}^\top&\bB_{[2]}^\top&\cdots&\bB_{[k]}^\top\end{bmatrix}^\top\in\R^{m\times n}$ have full column rank, with $\bB_{[i]}\in\R^{m_i\times n}$, and define $\blevsc{i}:=\Tr\left(\bB_{[i]}(\bB^\top\bB)^{-1}\bB_{[i]}^\top\right)$. Then:
\begin{enumerate}
\item $0\leq\blevsc{i}(\bB)\leq\rank(\bB_{[i]})$ for every $i\in[k]$.
\item $\sum_{i=1}^k\blevsc{i}(\bB)=n$.
\end{enumerate}
\end{restatable}

 We set $\omrp:=1-\frac{2}{p}\in(0,1)$ and define  $\Omega:=\{\bu\in\Rkpos:\one^\top\bu=n\}$, the set block weights live in. For $\bx,\by\in\Omega$, we define the scaled Kullback-Leibler (KL) divergence
\[
\kldiv{\bx}{\by} := \sum_{i=1}^k x_i\log\left(\frac{x_i}{y_i}\right).
\numberthis\label{eq:notation-KLdiv}
\]
For probability distributions $\mu$ and $\nu$, we use $\KL{\mu}{\nu}$ for the usual KL divergence. 

Given $\bw\in\R^k_{>0}$ and a specific block partitioning, let $\bW$ denote its block weight matrix as described above, and define  $\bG(\bw):=\bA^\top\bW^\omrp\bA$. Since $\bA$ has full column rank, $\bG(\bw)\succ\bzero$. We define  the ``Lewis map'' $\calT:\R_{>0}^k\to\Rkpos$ by: 
\[
\calT_i(\bw) := \blevsc{i}\left(\bW^{\omrp/2}\bA\right) = w_i^\omrp \Tr\left(\bAblk{i}\bG(\bw)^{-1}\bAblk{i}^\top\right)
\text{ for all }i\in[k]. \tag{Lewis map}
\]By \Cref{fact:block-levsc} and the assumption that every block is nonzero, $\calT(\Omega)\subseteq\Omega$. The true block Lewis weights likewise satisfy $\bwbar=\calT(\bwbar)\in\Omega$. When we specialize to row-wise Lewis weights in \Cref{sec:intuition}, we overload $\Omega$ with $\{\bu\in\Rmspos:\one^\top\bu=n\}$ and $\calT$ for the row-wise Lewis map $\calT_i(\bw):=\levsc{i}\left(\bW^{\omrp/2}\bA\right)$.

\section{Algorithm and analysis}
\label{sec:main-results}
In this section, we describe our algorithm and analysis that guarantee \Cref{thm:blkLW,thm:approx-levscores-blkLW}. Our algorithm and analysis for \Cref{thm:approx-levscores-blkLW} are a simple extension of those corresponding to \Cref{thm:blkLW}, and we present them in \Cref{{sec:approxLevScores}} after the case with exact leverage scores in \Cref{sec:computingViaExacts}.

\subsection{Computing via exact leverage scores}\label{sec:computingViaExacts} Our algorithm (displayed in \Cref{alg:main}) is a slight modification of \cite[Algorithm 6]{lee2016faster}. It takes the input matrix $\bA\in\R^{m\times n}$, with the given $k$ block partitions, and  first computes $r_i:=\rank(\bAblk{i})$ and $R:=\sum_{i=1}^k r_i$. When $R=n$, the  exact weights are $\wbar_i=r_i$, which the algorithm immediately returns. Otherwise, it initializes $w_i^{(0)}:=\frac{nr_i}{R}$ and runs $N=\left\lceil p\log\left(\frac{8p\sqrt{R}}{\eps}\right)\right\rceil$ iterations, each of whose key step is the following, which we refer to as the ``Lewis update'':\[ w_i^{(t+1)}:= \calT_i(\bw^{(t)}) =\blevsc{i}\left(\left(\bW^{(t)}\right)^{1/2-1/p}\bA\right), \text{ for all } i \in [k] \numberthis\label{eq:Lewis-update}.\]After $N$ iterations, the algorithm outputs $\what_i
:=
\left(
\frac{w_i^{(N)}}{(w_i^{(N-1)})^{1-2/p}}
\right)^{p/2}$,
for all $i\in[k].$ Since $\bw^{(0)}\in\Omega$ and $\calT(\Omega)\subseteq\Omega$, every iterate $\bw^{(t)}$, for $t =1, 2, \dotsc, N$, of \Cref{alg:main} satisfies $\bw^{(t)}\in\Omega$. 

In words, \Cref{eq:Lewis-update} rescales the rows of $\bA$ according to the current block weights $\bw^{(t)}$, computes the exact row leverage scores of 
the rescaled matrix $\left(\bW^{(t)}\right)^{1/2-1/p}\bA$, sums them within each block, and replaces each block weight for the next iteration by the corresponding block leverage score. 
\Cref{eq:Lewis-update} is motivated by the defining fixed-point relation $\wbar_i=\blevsc{i}(\bWbar^{1/2-1/p}\bA)$ for every $i\in[k]$ for the true block Lewis weights $\bwbar$. For the final output, we simply use the identity
$\wbar_i^{2/p}=\frac{\wbar_i}{\wbar_i^{1-2/p}}$.

As alluded to in \Cref{{sec:intro-contribution}}, the key step of our analysis  is to show that one Lewis update contracts a KL divergence to $\bwbar$ by the factor $1-\frac{2}{p}$ (\Cref{sec:kl-decay}). Iterating this update brings $\bG(\bw)$ sufficiently close spectrally  to $\bG(\bwbar)$, and combining this spectral closeness with the final output rule above yields  coordinatewise $(1\pm\eps)$-approximation of $\bwhat$ to $\bwbar$ (\Cref{sec:coordinate-wise-control}).

\begin{algorithm}[t]
\caption{Fixed-point iteration}
\label{alg:main}
\begin{algorithmic}[1]
\Require Full-column-rank matrix $\bA\in\R^{m\times n}$ partitioned into $k$ given nonzero  blocks $\bAblk{1}, \bAblk{2}, \dotsc, \bAblk{k}$ with $\bAblk{i}\in\R^{m_i\times n}$ for each $i\in[k]$, parameter $p>2$, and accuracy  $\eps\in(0, 1)$. 
\State Compute $r_i:=\rank(\bAblk{i})$ for all $i\in [k]$, and define $R:= \sum_{i = 1}^k r_i$. 
\If{$R=n$}
  \State \label{line:r-equals-n-case}\Return $\what_i=r_i$ for all $i\in[k]$. 
\EndIf
\State Initialize $w_i^{(0)}:= \frac{nr_i}{R}$ for all $i\in[k]$, and set $N = \left\lceil p \log\left(\frac{8p\sqrt{R}}{\eps}\right)\right\rceil$. \label{line:init}
\For{$t=0,1,\ldots,N-1$}
  \State Set $\bW^{(t)}:=\Diag\left(w_1^{(t)}\bI_{m_1}, w_2^{(t)}\bI_{m_2}, \dotsc, w_k^{(t)}\bI_{m_k}\right)$.
  \State Compute $\levsc\left((\bW^{(t)})^{1/2-1/p}\bA\right)\in\R^m_{\geq0}$, the  exact row-wise leverage scores of  $(\bW^{(t)})^{1/2-1/p}\bA$. 
  \For{block $i=1,\ldots,k$}
    \State Set $w_i^{(t+1)}:= \blevsc{i}\left( (\bW^{(t)})^{1/2-1/p}\bA \right) $.
  \EndFor
\EndFor
\State \Return $\what_i:=\left(\frac{w_i^{(N)}}{(w_i^{(N-1)})^{1-2/p}}\right)^{p/2}$ for all $i\in [k]$.
\end{algorithmic}
\end{algorithm}

\subsubsection{The block Lewis update contracts KL divergence}
\label{sec:kl-decay}

\begin{lemma}
\label{lem:kl-contraction}
Let  $\bw\in\Omega$. Then  $\kldiv{\bwbar}{\calT(\bw)}
    \leq
    \omrp\kldiv{\bwbar}{\bw}.$
\end{lemma}

\begin{proof}
By  the definition of block leverage scores and $\bwbar=\calT(\bwbar)$, we have \[ \calT_i(\bw)=w_i^\omrp \Tr\left(\bAblk{i}\bG(\bw)^{-1}\bAblk{i}^\top\right), \qquad \wbar_i = \wbar_i^{\omrp} \Tr\left(\bAblk{i}\bG(\bwbar)^{-1}\bAblk{i}^\top\right).\] Therefore, by expanding out the definition in \Cref{eq:notation-KLdiv}, we have \[\kldiv{\bwbar}{\calT(\bw)} = \omrp \kldiv{\bwbar}{\bw} + \sum_{i = 1}^k \wbar_i \log\left(\frac{\Tr\left(\bAblk{i}\bG(\bwbar)^{-1}\bAblk{i}^\top\right)}{\Tr\left(\bAblk{i}\bG(\bw)^{-1}\bAblk{i}^\top\right)}\right). \numberthis\label{eq:KLcon-2}\] We  complete the  proof  by showing next that the second term of \Cref{eq:KLcon-2} is nonpositive. To this end, we first note that by Jensen's inequality applied to $x\mapsto\log(x)$ and the fact that $\sum_{i=1}^k \wbar_i = n$, we have  \[ \sum_{i = 1}^k \frac{\wbar_i}{n} \log\left(\frac{\Tr\left(\bAblk{i}\bG(\bwbar)^{-1}\bAblk{i}^\top\right)}{\Tr\left(\bAblk{i}\bG(\bw)^{-1}\bAblk{i}^\top\right)}\right) \leq \log\left( \sum_{i=1}^k \frac{\wbar_i}{n} \frac{\Tr\left(\bAblk{i}\bG(\bwbar)^{-1}\bAblk{i}^\top\right)}{\Tr\left(\bAblk{i}\bG(\bw)^{-1}\bAblk{i}^\top\right)} \right),\] and it  suffices to show that \[\sum_{i=1}^k {\wbar_i} \frac{\Tr\left(\bAblk{i}\bG(\bwbar)^{-1}\bAblk{i}^\top\right)}{\Tr\left(\bAblk{i}\bG(\bw)^{-1}\bAblk{i}^\top\right)} \stackrel{?}{\leq} n.\numberthis\label{eq:KLcon-3}\]To this end, we use $\bwbar=\calT(\bwbar)$ and apply matrix Cauchy-Schwarz inequality $\left|\Tr\left(\bX^\top\bY\right)\right|\leq \|\bX\|_{\mathrm{F}}\|\bY\|_{\mathrm{F}}$ with $\bX:=\bG(\bw)^{1/2}\bG(\bwbar)^{-1}\bAblk{i}^\top$ and $\bY:= \bG(\bw)^{-1/2}\bAblk{i}^\top$ and first obtain \[\wbar_i \frac{\Tr\left(\bAblk{i} \bG(\bwbar)^{-1}\bAblk{i}^\top\right)}{\Tr\left(\bAblk{i} \bG(\bw)^{-1}\bAblk{i}^\top\right)} = \wbar_i^{\omrp} \frac{\left[\Tr \left(\bAblk{i}\bG(\bwbar)^{-1}\bAblk{i}^\top\right)\right]^2}{\Tr\left(\bAblk{i}\bG(\bw)^{-1}\bAblk{i}^\top\right)} \leq \wbar_i^{\omrp} \Tr\left(\bAblk{i} \bG(\bwbar)^{-1}\bG(\bw)  \bG(\bwbar)^{-1} \bAblk{i}^\top\right).\] Summing the above inequality over $i = 1, 2, \dotsc, k$ and applying the definition of $\bG(\bwbar)$ and $\bG(\bw)$, linearity of trace, and the invariance of trace to cyclic permutation, we have
\begin{align*}
\sum_{i=1}^k \wbar_i\frac{\Tr\left(\bAblk{i}\bG(\bwbar)^{-1}\bAblk{i}^\top\right)}{\Tr\left(\bAblk{i}\bG(\bw)^{-1}\bAblk{i}^\top\right)} &\leq \sum_{i=1}^k \wbar_i^\omrp \Tr\left(\bAblk{i}\bG(\bwbar)^{-1}\bG(\bw) \bG(\bwbar)^{-1} \bAblk{i}^\top\right)= \Tr\left(\bG(\bw)\bG(\bwbar)^{-1}\right).\numberthis\label{eq:KLcon-4}
\end{align*}
Next, again by $\bwbar=\calT(\bwbar)$, H\"older's inequality with exponents $\frac{1}{\omrp}$ and $\frac{1}{1-\omrp}$ (valid since we restrict ourselves to $p>2$), and the fact that $\bw\in\Omega$, we have \[ \Tr\left(\bG(\bw)\bG(\bwbar)^{-1}\right)
= \sum_{i=1}^k w_i^\omrp \Tr\left( \bAblk{i}\bG(\bwbar)^{-1} \bAblk{i}^\top\right) = \sum_{i=1}^k w_i^\omrp \wbar_i^{1-\omrp}\leq \left(\sum_{i=1}^k w_i\right)^\omrp \left(\sum_{i=1}^k \wbar_i\right)^{1-\omrp} = n. \numberthis\label{eq:KLcon-5}\]  Combining \Cref{eq:KLcon-4} and \Cref{eq:KLcon-5} completes the proof of \Cref{eq:KLcon-3}, which finishes the proof of the claim. 
\end{proof}

\begin{corollary}\label{cor:kl-decay}Assume $R>n$. 
Then for every $t\geq0$, the iterates  of \Cref{alg:main} satisfy \[\kldiv{\bwbar}{\bw^{(t)}}\leq\omrp^t n \log\left(\frac{R}{n}\right).\]
\end{corollary}
\begin{proof}
The assumption $R>n$ implies \Cref{alg:main} does not terminate at \Cref{line:r-equals-n-case} and instead initializes $w_i^{(0)} = \frac{nr_i}{R}$. We also know from \Cref{fact:block-levsc} that $0\leq \wbar_i \leq r_i$ and $\sum_{i=1}^k \wbar_i = n$. Therefore, \[\kldiv{\bwbar}{\bw^{(0)}}= \sum_{i=1}^k \wbar_i\log\left(\frac{\wbar_i}{w^{(0)}_i}\right)\leq n\log\left(\frac{R}{n}\right). \numberthis\label{eq:cor-con-1}\]Applying \Cref{lem:kl-contraction} at $\bw = \bw^{(s)}$  gives $\kldiv{\bwbar}{\bw^{(s+1)}}= \kldiv{\bwbar}{\calT(\bw^{(s)})}\leq \omrp\kldiv{\bwbar}{\bw^{(s)}}$; repeatedly applying this at $s = 0, 1, 2, \dotsc, t-1$ and plugging in \Cref{eq:cor-con-1}  gives the claimed decay. 
\end{proof} 

\subsubsection{From KL contraction to a spectral bound}
\label{sec:coordinate-wise-control}
 For two compatible matrices $\bX, \bY\succ\bzero$, we define the Bregman divergence~\cite{bregman1967relaxation} \[\logdetdiv{\bX}{\bY}:= -\log\det\bX+\log\det\bY + \Tr\left(\bY^{-1}\left(\bX-\bY\right)\right).\] 
\begin{lemma}\label{lem:normalized-gram-eigs-small}
Let $\bw\in \Omega$. Then, 
 we have $\logdetdiv{\bG(\bw)}{\bG(\bwbar)} \leq \omrp\kldiv{\bwbar}{\bw}$. 
\end{lemma}
\begin{proof}
Our proof strategy is as follows: We will show that \[\Tr\left(\bG(\bw)\bG(\bwbar)^{-1}\right)\stackrel{?}{\leq} n, \qquad \log\det(\bG(\bwbar)) - \log\det(\bG(\bw)) \stackrel{?}{\leq}\omrp\kldiv{\bwbar}{\bw}.\numberthis\label{eq:spectral-1}\] Using these bounds in the definition of $\logdetdiv{\bG(\bw)}{\bG(\bwbar)}$ gives the below chain of steps: \begin{align*}\logdetdiv{\bG(\bw)}{\bG(\bwbar)} &= \Tr\left(\bG(\bw)\bG(\bwbar)^{-1}\right) - n + \log\det(\bG(\bwbar)) - \log\det(\bG(\bw))\\ & \leq \omrp\kldiv{\bwbar}{\bw}.\end{align*} which  concludes the proof. We now show \Cref{eq:spectral-1}. We already have  $\Tr\left(\bG(\bw)\bG(\bwbar)^{-1}\right)\leq n$  from \Cref{eq:KLcon-5}. To prove the second inequality in \Cref{eq:spectral-1}, we start with the fact that $\bx\mapsto\log\det\bG(e^{\bx})$ is convex (cf. \Cref{lem:logdetGexp-is-cvx}). Then, by the linear underestimator property of convex functions and the identities $\frac{\partial}{\partial x_i}\log\det\bG\left(e^{\bx}\right) = \omrp \calT_{i}\left(e^{\bx}\right)$ (cf. \Cref{lem:logdetGexp-is-cvx}) and $\bwbar=\calT(\bwbar)$,  we have \begin{align*}
\log\det\bG(\bw) &= \log\det\bG\left(e^{\log \bw}\right)\\ &\geq \log\det\bG\left(e^{\log \bwbar}\right) + \nabla_{\bx}\log\det\bG\left(e^{\bx}\right)|_{\bx=\log\bwbar}^\top (\log \bw-\log\bwbar)\\  
&= \log\det\bG(\bwbar) + \omrp\cdot \sum_{i = 1}^k \wbar_i \log\left(\frac{w_i}{\wbar_i}\right) = \log\det\bG(\bwbar) - \omrp \kldiv{\bwbar}{\bw}. 
\end{align*} Rearranging the terms  proves the second inequality in \Cref{eq:spectral-1}, thereby proving the lemma. 
\end{proof}

\begin{lemma}
\label{lem:coordinatewise-approx}
Let $\varepsilon\in (0, 1)$, and let $\bw\in\Omega$ satisfy $\logdetdiv{\bG(\bw)}{\bG(\bwbar)}\leq \frac{\eps^2}{64p^2}$. Define $\what_i := \left(\frac{\calT_i(\bw)}{w_i^\omrp}\right)^{\frac{p}{2}}$ for all $i \in [k]$. Then, $(1-\eps)\wbar_i \leq \what_i \leq (1+\eps)\wbar_i$ for all $i \in [k]$.
\end{lemma}
\begin{proof}
Set $\delta := \frac{\eps^2}{64p^2}$ and $a:= 2\sqrt{\delta}$. Then, under the premise of our lemma (and using our standing assumption that $p>2$), \cite[Lemma~20]{pmlr-v336-gribling26a} yields: \[(1-a)\bG(\bwbar)\preceq \bG(\bw)\preceq (1+a) \bG(\bwbar).\] Taking inverses, pre- and post-multiplying  by appropriate matrices, and taking the trace yields: \[\frac{1}{1+a}\Tr\left(\bAblk{i}\bG(\bwbar)^{-1}\bAblk{i}^\top\right)\leq \Tr\left(\bAblk{i}\bG(\bw)^{-1}\bAblk{i}^\top\right)\leq \frac{1}{1-a}\Tr\left(\bAblk{i}\bG(\bwbar)^{-1}\bAblk{i}^\top\right).\numberthis\label{eq:coord1}\] Since $\bwbar= \calT(\bwbar)$, we have $\Tr\left(\bAblk{i}\bG(\bwbar)^{-1}\bAblk{i}^\top\right) = \wbar_i^{2/p}.$ Substituting this back into \Cref{eq:coord1}, exponentiating by $p/2$, and plugging in the definitions of $\bwbar$ and $\widehat{\bw}$ gives
\[(1+a)^{-p/2}\wbar_i\leq \what_i \leq (1-a)^{-p/2}\wbar_i.\numberthis\label{eq:coord2}\]
We simplify both sides of the above inequality chain, one at a time. First, $(1+a)^{-p/2} \geq e^{-ap/2} = e^{-\varepsilon/8}\geq 1-\varepsilon/8\geq 1-\varepsilon.$ Next, since $a<1/2$, we have $-\log(1-a)\leq 2a$, and therefore $(1-a)^{-p/2} \leq e^{ap} =e^{\eps/4} \leq 1+\eps$.  Plugging these simplifications in \Cref{eq:coord2} yields the claim.  
\end{proof}
\subsubsection{Proof of \Cref{thm:blkLW}}\label{sec:main-proof}
\begin{proof}[Proof of \Cref{thm:blkLW}]
First, we note that by subadditivity of rank under vertical concatenation, it is always  the case that $R:= \sum_{i=1}^k \rank\left(\bAblk{i}\right)\geq \rank(\bA)=: n$. 
When $R = n$, one can infer from \Cref{fact:block-levsc} that each $\blevsc{i}=\rank\left(\bAblk{i}\right)$ for all $i\in[k]$, and, as designed in \Cref{alg:main}, the block Lewis weights are  $\what_i=\rank\left(\bAblk{i}\right)$ for all $i\in[k]$, thereby proving correctness in this case. The claimed iteration bound is immediate in this case, since \Cref{alg:main} computes no leverage scores.  
We now assume the other case, $R>n$.
By \Cref{lem:normalized-gram-eigs-small} applied to $\bw= \bw^{(N-1)}$, \Cref{cor:kl-decay} applied to $t=N-1$, the choice of $\omrp=1-\frac{2}{p}$, and the choice of $N = \left\lceil p\log\left(\frac{8p\sqrt{R}}{\eps}\right)\right\rceil$ in \Cref{alg:main}, we have 
\[ \logdetdiv{\bG(\bw^{(N-1)})}{\bG(\bwbar)}\leq \omrp \kldiv{\bwbar}{\bw^{(N-1)}}\leq \gamma^N n \log\left(\frac{R}{n}\right)\leq R e^{-2N/p}\leq \frac{\eps^2}{64p^2}.\]\looseness=-1 Thus, the premise of \Cref{lem:coordinatewise-approx} is satisfied at $\bw=\bw^{(N-1)}$; further, the output of \Cref{alg:main} is exactly $\bwhat$ from \Cref{lem:coordinatewise-approx} with $\bw=\bw^{(N-1)}$; hence, \Cref{lem:coordinatewise-approx} implies $(1-\eps)\wbar_i \leq \what_i\leq (1+\eps)\wbar_i$ for all $i\in [k]$ for the output of \Cref{alg:main}. Each of the $N$ iterations of \Cref{alg:main} computes  one exact  leverage score vector and  sums the entries within each block, implying the claimed  complexity. 
\end{proof}

\begin{remark}\label{rem:rank_not_needed}
    \Cref{alg:main} requires computing the exact rank. However, we can  completely avoid this numerically unstable step and replace $r_i$ with $m_i$ throughout \Cref{alg:main}, which mildly worsens the iteration count in \Cref{thm:blkLW} from $O\left(p\log\left(\frac{p\sqrt{R}}{\varepsilon}\right)\right)$ to $O\left(p\log\left(\frac{p\sqrt{m}}{\varepsilon}\right)\right)$. 
\end{remark}

\subsection{Computing via approximate leverage scores}\label{sec:approxLevScores}
\input{approxlevscores}

\section{Two additional perspectives on the Lewis fixed-point iteration}\label{sec:intuition}
\looseness=-1By \Cref{sec:main-results}, the update in \Cref{eq:Lewis-update} contracts the KL divergence to the true Lewis weights by $1-\frac{2}{p}$ at each iteration, driving the $\widetilde{O}(p)$ rate in \Cref{thm:blkLW}. We give two additional perspectives on this iteration.

\subsection{Global dynamics: a volume sampling perspective}\label{sec:global-intuition}

First, we present the volume sampling perspective on the \emph{row-wise} Lewis update $\bw\mapsto\calT(\bw)$. The main idea is: 
Under the external tilt $\bw^{1-2/p}$, the marginal inclusion probability vector of the resulting volume sampling distribution is exactly $\calT(\bw)$. 
Combining this with entropic independence~\cite{anari2022entropic} gives an alternate proof of the KL contraction in \Cref{lem:kl-contraction} and explains the appearance of the factor $1-\frac{2}{p}$.
Below, we elaborate on this argument in the row-wise setting; the same argument extends to blocks by grouping the row marginals within each block.

  For any $\bv\in\Rmspos$, let $\pi_{\bv}$ denote the externally tilted volume sampling distribution 
  over $n$-element subsets of the $m$ rows of $\bA$:
  \[\pi_{\bv}(S) := \frac{(\det(\bA_S))^2\prod_{i \in S} v_i}{Z(\bv)},\, S\subseteq[m], |S| = n; \quad Z(\bv) := \sum_{\substack{S\subseteq[m],\\|S|=n}} \det(\bA_S)^2\prod_{i\in S}v_i.\numberthis\label{eq:EI-con-0}\] 
Since $\det(\bA_S)^2\prod_{i\in S}v_i = \det\left(\left(\Diag(\sqrt{\bv})\bA\right)_S\right)^2$, one can equivalently describe $\pi_{\bv}$ as the volume sampling distribution induced by the matrix $\Diag(\sqrt{\bv})\bA$. Under this $n$-element volume sampling distribution, it is known that the marginal probability that row $i$ is included equals the $i^\mathrm{th}$ leverage score of $\Diag(\sqrt{\bv})\bA$~\cite[Proposition 4]{derezinski2018reverse}\cite[Theorem 8]{derezinski2021determinantal}, i.e., \[ \Prob_{S\sim \pi_{\bv}}\left(i\in S\right) = \levsc{i}\left(\Diag\left(\sqrt{\bv}\right)\bA\right), \, \text{ for all } i\in[m]. \numberthis\label{eq:EI-con-01}\] We now specialize \Cref{eq:EI-con-01} to $\bv=\bw^{\omrp}$ and $\bv=\bwbar^\omrp$ and use $\Diag\left(\sqrt{\bw^\omrp}\right)\bA = \bW^{\omrp/2}\bA$ and the fixed point characterization $\bwbar=\calT(\bwbar)$  to conclude the following facts we shortly use: \begin{equation}\begin{aligned}\Prob_{S \sim \pi_{\bw^\omrp}}\left(i\in S\right) &= \levsc{i}\left(\bW^{\omrp/2}\bA\right) = \calT_i(\bw), \, \text{ for all } i\in[m]. \\
\Prob_{S\sim\pi_{\bwbar^\omrp}}\left(i\in S\right) &= \bwbar_i, \, \text{ for all } i\in[m]. \label{eq:EI-con-02}\end{aligned}\end{equation} 
We now use the above established facts. 
To this end, we invoke entropic independence, a strong
data-processing property, introduced by Anari, Jain, Koehler, Pham, and Vuong~\cite{anari2022entropic},  for distributions on fixed-size sets: Let $\mu$ and $\nu$ be two distributions on $n$-element subsets of $[m]$, with strictly positive marginal inclusion probability vectors $\ba$ and $\bb$, respectively. If, after sampling a set, we retain only one uniformly chosen random element, the resulting distributions on $[m]$ 
are $\frac{\ba}{n}$ and $\frac{\bb}{n}$. The standard data-processing inequality  implies that this operation cannot increase KL divergence;  if, however, $\nu$ is $1$-entropically independent, then the following stronger bound holds~\cite{anari2022entropic}: 
\[ 
\KL{\frac{\ba}{n}}{\frac{\bb}{n}} \leq \frac{1}{n}\KL{\mu}{\nu}. 
\]
 Equivalently, using
\Cref{eq:notation-KLdiv} on the left-hand side, we have
\[
\kldiv{\ba}{\bb}\leq \KL{\mu}{\nu}.
\numberthis\label{eq:EI-con-1}
\]In our context, fix $\bw\in\Omega$ and set
$\mu=\pi_{\bwbar^\omrp}$ and $\nu=\pi_{\bw^\omrp}$.
The distribution $\nu=\pi_{\bw^\omrp}$ is a determinantal volume sampling measure and hence  $1$-entropically independent~\cite{anari2022entropic,anari2025optimal}.
Therefore, \Cref{eq:EI-con-1} applies.
By \Cref{eq:EI-con-02}, the marginal vectors of $\mu$ and $\nu$ are $\bwbar$ and $\calT(\bw)$, respectively, and hence
\[
\kldiv{\bwbar}{\calT(\bw)}
\leq
\KL{\pi_{\bwbar^\omrp}}{\pi_{\bw^\omrp}}.
\numberthis\label{eq:EI-con-2}
\] We now expand the right-hand side above: \begin{align*}\KL{\pi_{\bwbar^\omrp}}{\pi_{\bw^\omrp}} &=  \E_{S\sim \pi_{\bwbar^\omrp}}\left(\log\left(\frac{\pi_{\bwbar^\omrp}(S)}{\pi_{\bw^\omrp}(S)}\right)\right)  \\  
&= \omrp\E_{S\sim \pi_{\bwbar^\omrp}} \left( \sum_{i \in S}\log\left(\frac{\wbar_i}{w_i}\right)\right)  + \log\left(\frac{Z(\bw^\omrp)}{Z(\bwbar^\omrp)}\right) \\ 
&= \omrp\sum_{i=1}^m \Prob_{S\sim \pi_{\bwbar^\omrp}}(i\in S) \log\left(\frac{\wbar_i}{w_i}\right) + \log\left(\frac{Z(\bw^\omrp)}{Z(\bwbar^\omrp)}\right) \\ 
&= \gamma\sum_{i=1}^m \wbar_i \log\left(\frac{\wbar_i}{w_i}\right) + \log\left(\frac{Z(\bw^\omrp)}{Z(\bwbar^\omrp)}\right) \\
&= \gamma \kldiv{\bwbar}{\bw} +\log\det\bG(\bw)-\log\det\bG(\bwbar),\numberthis\label{eq:EI-con-3}
\end{align*} where the first step is by the definition of the usual KL divergence, the second step by \Cref{eq:EI-con-0}, the third step by representing membership in $S$ via  indicator variables and using linearity of expectations, the fourth step by \Cref{eq:EI-con-02}, and the fifth step uses Cauchy-Binet \Cref{{eq:logdetG-1}} and our definition of $\bG$ to conclude $Z(\bw^\omrp) = \det\left(\bA^\top\bW^\omrp\bA\right)=\det(\bG(\bw))$. Finally, by the variational characterization of Lewis weights~\cite{cohen2015lp,lee2019solving}, we have $\bwbar\in\arg\max_{\bu\in\Omega}\log\det\bG(\bu),$ which simplifies 
\Cref{eq:EI-con-3} to \[ \KL{\pi_{\bwbar^\omrp}}{\pi_{\bw^\omrp}} \leq \omrp \kldiv{\bwbar}{\bw}.\numberthis\label{eq:EI-con-4}\] Chaining together \Cref{eq:EI-con-2} and \Cref{eq:EI-con-4}  yields $\kldiv{\bwbar}{\calT(\bw)}\leq \left(1-\frac{2}{p}\right)\kldiv{\bwbar}{\bw}$, thereby recovering the  contraction in \Cref{lem:kl-contraction} in the row-wise setting. 

Thus, the above calculations explain the $1-\frac{2}{p}$ factor: the right-hand side in \Cref{eq:Lewis-update} corresponds to the external-field tilt $\bw^{\omrp}$, and upon expanding the KL divergence between the two tilted volume sampling distributions, this exponent appears as the coefficient of $\kldiv{\bwbar}{\bw}$.

\subsection{Local dynamics: linearization at the fixed point}
\label{sec:local-dyn}
The above volume sampling perspective explains the \emph{global} KL contraction in \Cref{lem:kl-contraction} and the associated $1-\frac{2}{p}$ factor; however, this contraction does not characterize the asymptotic convergence rate of the iterates near the fixed point $\bwbar$. Motivated by stable geometric decay we observe in exploratory numerical experiments, we develop a more precise theory of these \emph{local} dynamics. 
To this end, we linearize the row-wise version of the update \Cref{eq:Lewis-update} around $\bwbar$; the spectral radius of the resulting Jacobian predicts the asymptotic local contraction rate. 
As seen in \Cref{fig:local-dyn}, this predicted rate is closely matched by the observed dynamics on independently generated test instances.

\begin{lemma}\label{lem:local-dyn} Let  $\wbar\in\Rmspos$ be the row-wise Lewis weights of $\bA$, and let $\calT(\bw):=\levsc\left(\bW^{\omrp/2}\bA\right)$ denote \Cref{eq:Lewis-update} applied row-wise. Define the  matrices $\bP:= \bWbar^{\omrp/2}\bA\left(\bA^\top \bWbar^\omrp \bA\right)^{-1}\bA^\top\bWbar^{\omrp/2}\in\R^{m\times m}$, and $\bK:= \bWbar^{-1}(\bP\circ\bP) \in\R^{m \times m}$. Define the Jacobian $\bJ:= \left.\left[\frac{\partial}{\partial x_j}\log\calT_i(e^{\bx})\right]\right|_{i, j\in [m]}\lvert_{\bx = \log\bwbar}$. Then, $\bJ = \omrp(\bI-\bK).$ Moreover, all eigenvalues of $\bK$ lie in $[0,1]$. Therefore, the spectral radius of $\bJ$ is
\[ \mu_{\mathrm{pred}} = \left(1-\frac{2}{p}\right) \left(1-\lambda_{\min}(\bK)\right).
\]
\end{lemma}
\begin{proof}
For this proof, we introduce the shorthand notation  $\bq(\bx):= e^{\omrp \bx/2}$,   $\bQ(\bx):= \Diag(\bq(\bx))$, and $\bP(\bx):= \bQ(\bx)\bA\left(\bA^\top\bQ(\bx)^2\bA\right)^{-1}\bA^\top \bQ(\bx)$. Then, it follows that \[\diag(\bP(\bx))= \levsc(\bQ(\bx)\bA) = \calT(e^{\bx}).\numberthis\label{eq:locdyn-0}\] Fix $\bx\in\R^m$ so that $\bq(\bx)\in\R^m$ is fixed too. Then, specializing \cite[Lemma~49]{lee2019solving} to $\bQ(\bx)\bA\in\R^{m\times n}$ and using \Cref{eq:locdyn-0} gives, for any perturbation $\bu\in\R^m$, \[\frac{d}{dt} \levsc\left(\Diag\left(\bq(\bx) + t\bu\right)\bA\right)|_{t=0} = 2\left(\Diag(\calT(e^{\bx})) - \bP(\bx)\circ\bP(\bx) \right)\bQ(\bx)^{-1}\bu.\]Next, perturbing $\bx$ in the direction $\bh$ gives $\frac{d}{dt}\bq(\bx+t\bh)|_{t=0} = \frac{\omrp}{2}\bQ(\bx)\bh.$ 
Therefore,  chain rule gives \begin{align*}\frac{d}{dt}\calT(e^{\bx+t\bh})|_{t=0} &= 2\left(\Diag(\calT(e^{\bx})) - \bP(\bx)\circ\bP(\bx)\right)\bQ(\bx)^{-1}\left(\frac{\omrp}{2}\bQ(\bx)\bh\right)\\ 
&= \omrp\cdot\left(\Diag(\calT(e^{\bx})) - \bP(\bx)\circ\bP(\bx)\right)\bh. 
\end{align*} Again, by chain rule, we have \begin{align*}
\frac{d}{dt}\log\calT(e^{\bx+t\bh})|_{t=0} &= \Diag(\calT(e^{\bx}))^{-1}\frac{d}{dt}\calT(e^{\bx+t\bh})|_{t=0}\\ 
&= \omrp\left(\bI- \Diag(\calT(e^{\bx}))^{-1}\left(\bP(\bx)\circ\bP(\bx)\right)\right)\bh.\numberthis\label{eq:locdyn-2}\end{align*} Setting $\bx=\log\bwbar$ and using the fact that $\bwbar=\calT(\bwbar)$ implies $\bP(\log\bwbar)= \bP$; plugging these, along with the definition of $\bK$, in \Cref{eq:locdyn-2} gives the claimed Jacobian. 

Next, we bound the spectrum of $\bK$. By Part (2) of \cite[Lemma~47]{lee2019solving}, we have $\bzero \preceq \bP\circ\bP \preceq \bWbar.$
Pre and post multiplying by $\bWbar^{-1/2}$, this simplifies to 
\[\bzero \preceq \bWbar^{-1/2} (\bP\circ\bP) \bWbar^{-1/2} \preceq \bI.
\] Finally, note that $\bK$ is similar to the above $\bWbar^{-1/2}(\bP\circ\bP)\bWbar^{-1/2}$, thus implying $0\leq \lambda_i(\bK)\leq 1$ for all $i = 1, 2, \dotsc, m$. This immediately yields the claimed spectral radius, thus finishing the proof. 
\end{proof}

\noindent Thus, unlike the global KL factor $1-\frac{2}{p}$, the local contraction rate depends on the local geometry and can be strictly faster. We now explain how we test this theory in our experiments. 

\paragraph{Prediction for the measured residual.} We next turn the local prediction in \Cref{lem:local-dyn} into one for a quantity that we can measure experimentally. Let $\bx^{(t)}:=\log\bw^{(t)}$ and $\bxbar:=\log\bwbar$, and define the residual
\[R_t:=\max_{i\in[m]}\left|\log\frac{w_i^{(t+1)}}{w_i^{(t)}}\right|=\left\|\bx^{(t+1)}-\bx^{(t)}\right\|_\infty.\]
The decay of $R_t$ immediately relates to the Jacobian $\bJ$ from \Cref{lem:local-dyn}: By the definition of $\bJ$ and a Taylor expansion of the map $\bx\mapsto\log\calT(e^{\bx})$ around $\bxbar$, as the iterates approach the fixed point we have
\[\bx^{(t+2)}-\bx^{(t+1)} = \bJ\left(\bx^{(t+1)}-\bx^{(t)}\right) + o\left(
\left\|\bx^{(t+1)}-\bx^{(t)}\right\|\right).\] Writing $\bx^{(t+1)}-\bx^{(t)}$ in an eigenbasis of $\bJ$, if its coefficient along an eigendirection corresponding to the largest eigenvalue $\rho(\bJ)$ is nonzero, then repeated application of $\bJ$ causes this slowest-decaying component to dominate. We therefore expect $\frac{R_{t+1}}{R_t}\longrightarrow\rho(\bJ)=\mu_{\mathrm{pred}}$. Thus, near the fixed point, we expect $R_t\approx C\mu_{\mathrm{pred}}^t$, so $\log R_t$ is approximately linear in $t$ with slope $\log\mu_{\mathrm{pred}}$. In \Cref{sec:expts-A0}, we estimate the observed contraction factor $\mu_{\mathrm{obs}}$ from this slope and compare it with $\mu_{\mathrm{pred}}$ (\Cref{fig:local-dyn}).

\section{Numerical experiments}
\label{sec:expts} 
We implemented our algorithm and tested it on synthetic instances and real data. Our code is available at \href{https://github.com/sw-ti/Lewis-weights-paper-code}{https://github.com/sw-ti/Lewis-weights-paper-code}.

\subsection{Synthetic experiments}\label{sec:syn_expt}
\looseness=-1In \Cref{sec:expts-A0}, we begin with a question about local dynamics: we ask if the asymptotic contraction observed near the Lewis fixed point agrees with the theory developed in \Cref{sec:local-dyn}. We then ask, in \Cref{{sec:A1-A2-A4}}, if the global $\widetilde{O}(p)$ dependence in \Cref{cor:ordinary-main} is reflected in the observed iteration counts of the \emph{row-wise} Lewis fixed-point iteration; we first test this  on i.i.d.\ Gaussian matrices and then on two matrix families that depart from this baseline: first on a structured, well-conditioned matrix family with rows at two significantly different scales of leverage scores, and then on ill-conditioned matrices. Finally, in \Cref{{sec:expt-blkLW}}, we  ask whether the \emph{block} Lewis fixed-point iteration exhibits the analogous qualitative dependence on $p$ suggested by \Cref{thm:blkLW}.

\subsubsection{Local dynamics}
\label{sec:expts-A0}

To test the local dynamics described in \Cref{sec:local-dyn}, we generate $30$ independent matrices $\bA\in\R^{100\times20}$ with i.i.d.\ standard Gaussian entries and repeatedly apply the row-wise specialization of \Cref{eq:Lewis-update} for $14$ values of $p$ between $3$ and $64$. We measure progress by
$R_t:=\max_{i\in[m]}\left|\log\left(\frac{w_i^{(t+1)}}{w_i^{(t)}}\right)\right|$, the largest (across all rows) absolute change in log-weights between successive iterates; thus $R_t=0$ exactly at a fixed point, while $R_t>0$ quantifies the maximum  violation of the fixed-point relation.

For each pair of tested matrix and $p$, \Cref{lem:local-dyn} predicts that the asymptotic contraction factor is $\mu_{\mathrm{pred}}=\left(1-\frac{2}{p}\right)(1-\lambda_{\min}(\bK))$. We estimate the observed contraction factor $\mu_{\mathrm{obs}}$ by fitting $\log R_t$ against $t$ over the prespecified window $10^{-8}\leq R_t\leq10^{-4}$. This window is chosen  so as to exclude both early iterations  too far from the fixed point for the local linearization to accurately approximate the dynamics as well as later iterations where floating-point noise may dominate the residual. This gives $30\times14=420$ datapoints to compare the predictions and corresponding observations.

\begin{figure}[!htbp]
    \centering
    \includegraphics[width=0.95\textwidth]{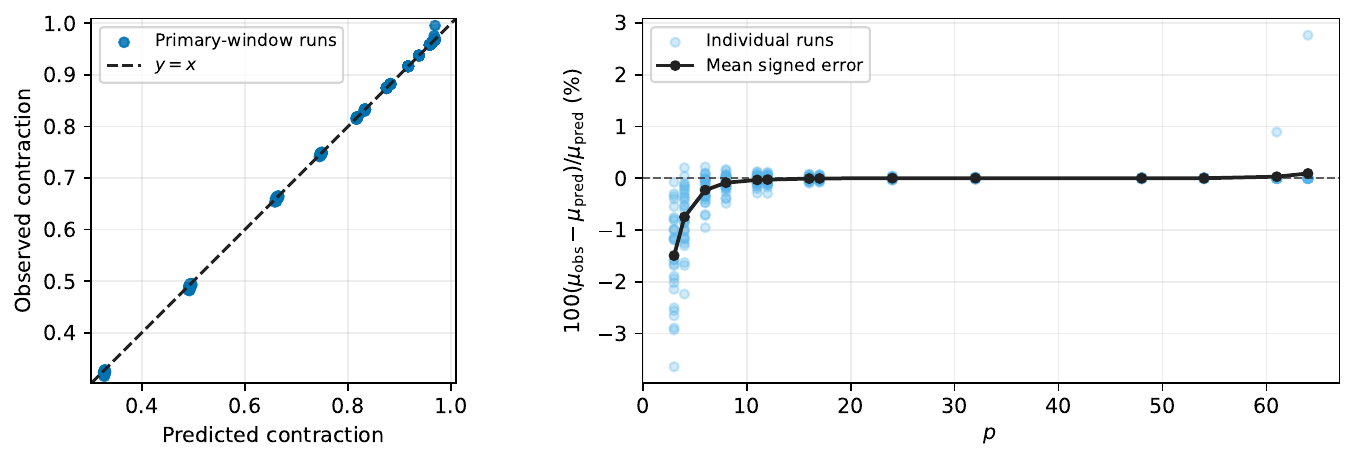}
    \caption{Predicted and observed local contraction factors for our Lewis weight iteration on Gaussian matrices.}
    \label{fig:local-dyn}
\end{figure}

\looseness=-1We see that our predicted and observed contraction factors closely agree across the tested range of $p$ (\Cref{fig:local-dyn}). In $418$ of the $420$ runs, $\log R_t$ is nearly a straight line in $t$ over the primary fitting window, as expected under geometric decay. In the remaining two runs, at $p=61$ and $p=64$ for the same matrix,  the residual fluctuates irregularly near $10^{-8}$, which we attribute to floating-point noise. To test if the poor fit is due to these fluctuations, we repeat the analysis over the prespecified window $10^{-7}\leq R_t\leq10^{-4}$, removing the part with low residual; here, we find that $\log R_t$ is nearly linear in $t$ in all $420$ runs. Thus, the observed local dynamics empirically support the theory in \Cref{lem:local-dyn}.

\subsubsection{{Row-wise Lewis weights: iteration count and robustness}}\label{sec:A1-A2-A4}
\looseness=-1We now study the iteration counts and numerical robustness of the row-wise Lewis update.

\paragraph{Iteration count of row-wise Lewis weights.}
We first ask how the row-wise iteration count varies with $p$ on generic dense instances. We generate 50 fixed matrices $\bA\in\R^{100\times20}$ with i.i.d.\ standard Gaussian entries and reuse them across an irregular $p$-grid from $3$ to $127.3$. Starting from the uniform initialization, we use $R_t:=\max_{i\in[m]}\left|\log\frac{w_i^{(t+1)}}{w_i^{(t)}}\right|<10^{-8}$ as our primary stopping rule, and also record the first crossing of $10^{-4}$ and $10^{-6}$. For each $p$, the left panel of \Cref{fig:A1} summarizes the $10^{-8}$ stopping time across the 50 matrices by its median and empirical 20th and 80th percentiles, while the right reports the median stopping time at all three thresholds.

\begin{figure}[!htbp]
    \centering
    \includegraphics[width=0.95\textwidth]{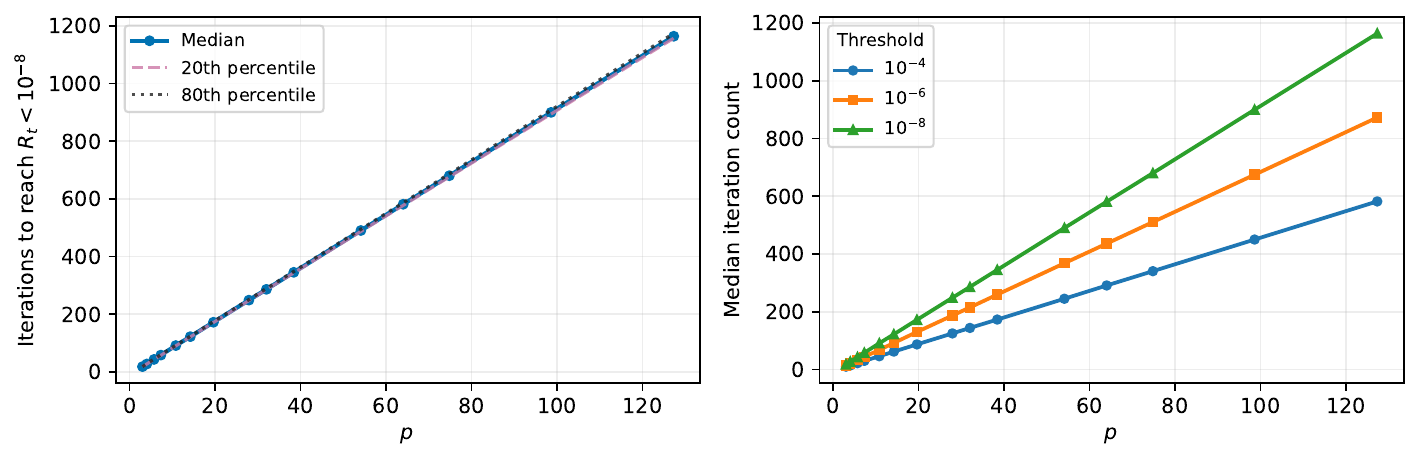}
    \caption{Row-wise Lewis weight iteration count versus $p$ (left); sensitivity to the stopping threshold (right). }
    \label{fig:A1}
\end{figure}

Across the tested $p$, the iteration count grows approximately linearly with $p$, varying little across instances; this qualitative behaviour persists at thresholds $10^{-4}$ and $10^{-6}$ (\Cref{fig:A1}). 
Thus, at fixed residual thresholds, our experiments are consistent with the $\widetilde O(p)$ dependence in \Cref{cor:ordinary-main}.

\paragraph{Structured row geometry.}
We next ask whether the behaviour in \Cref{fig:A1} is peculiar to generic Gaussian matrices.
For each of 50 instances, we construct $\bA = \begin{bmatrix} \bI_{20}\\[1mm] \frac{1}{4} \widetilde{\bG}   \end{bmatrix}\in \R^{100\times 20}$, where the 80 rows of $\widetilde{\bG}$ are independent Gaussian directions normalized to unit Euclidean norm. Thus the 20 ``anchor'' rows have norm one and the 80 ``cloud'' rows norm $\frac{1}{4}$. This creates leverage heterogeneity without poor conditioning: the median anchor/cloud leverage ratio is about $16.8$, while the spectral condition number $\kappa_2(\bA):=\frac{\sigma_{\max}(\bA)}{\sigma_{\min}(\bA)}$
ranges from roughly $1.16$ to $1.21$; see the right panel of \Cref{fig:structured-ordinary}.

\begin{figure}[!htbp]
    \centering
    \includegraphics[width=0.95\textwidth]{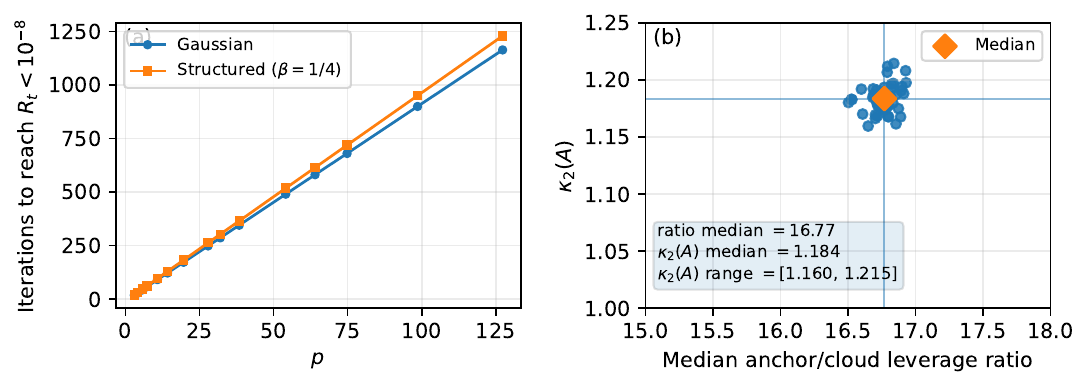}
    \caption{Row-wise Lewis weight iteration counts on Gaussian and structured matrices (left), and leverage heterogeneity versus conditioning for the structured family (right).}
    \label{fig:structured-ordinary}
\end{figure}

Using the same $p$-grid, initialization, and stopping rule as the Gaussian experiment (\Cref{fig:A1}), the left panel of \Cref{fig:structured-ordinary}  shows that the tested structured family closely tracks the Gaussian baseline across the full tested range of $p$, requiring approximately $5$-$7\%$ more iterations for $p\geq7.3$. Thus the same qualitative $p$-dependence observed on the Gaussian baseline persists under this leverage-heterogeneous --- but well-conditioned --- departure from generic matrices.

\paragraph{Numerical robustness under ill-conditioning.}
We next test the effect of ill-conditioning on  floating-point implementation. For each of five base instances $j\in[5]$ and nine parameters $s$ that determine the condition numbers, we construct $\bA_{j,s}=\bQ_j\bD_s\bV_j^\top$, where $\bQ_j\in\R^{100\times20}$ has orthonormal columns, $\bV_j\in\R^{20\times20}$ is orthogonal, and the diagonal entries of $\bD_s$ decrease geometrically; across these choices, spectral condition number $\kappa_2(\bA_{j,s})$ ranges from approximately $1$ to $10^{16}$. Since the Lewis map is invariant under invertible right transformations, these matrices all define the same Lewis map in exact arithmetic: writing $\bA_{j,0}=\bQ_j\bV_j^\top$, we have 
\[
    \bA_{j,s}=\bA_{j,0}(\bV_j\bD_s\bV_j^\top), \qquad \calT_{\bA_{j,s}}(\bw)=\calT_{\bA_{j,0}}(\bw).
\]Therefore, from the same initialization, exact arithmetic would produce the same Lewis weights sequence at every conditioning level, and any differences observed are finite-precision effects.

We measure these effects both at the level of one update and over the full trajectory. For fixed $p$, let $\bw_j^{(t)}(s)$ denote the $t^\mathrm{th}$ iterate for $\bA_{j,s}$, and let $T_j(0)$ denote the stopping time of the well-conditioned $s=0$ run. We record the following three quantities:
\[
    R_{j,t}(s):=\max_{i\in[m]}\left|\log\frac{w_{j,i}^{(t+1)}(s)}{w_{j,i}^{(t)}(s)}\right|, \quad     E_{j,t}(s):=\max_{i\in[m]}\left|\log\frac{[\calT_{\bA_{j,s}}(\bw_j^{(t)}(0))]_i}       {[\calT_{\bA_{j,0}}(\bw_j^{(t)}(0))]_i}
    \right|, \quad  E_{j,\max}(s):=\max_{0\le t<T_j(0)}E_{j,t}(s).
\]
Thus $R_{j,t}(s)$ measures progress along the run on $\bA_{j,s}$, and $E_{j,t}(s)$ compares the two floating-point maps at the reference iterate $\bw_j^{(t)}(0)$. Exact arithmetic gives $R_{j,t}(s)=R_{j,t}(0)$ and $E_{j,\max}(s)=0$. We reuse the primary $10^{-8}$ stopping tolerance from the row-wise scaling experiment (\Cref{fig:A1}) and test six values of $p$ spanning $4$ to $127.3$, for $270$ final trajectories. 
\begin{figure}[!htbp]
    \centering
    \includegraphics[width=0.95\textwidth]{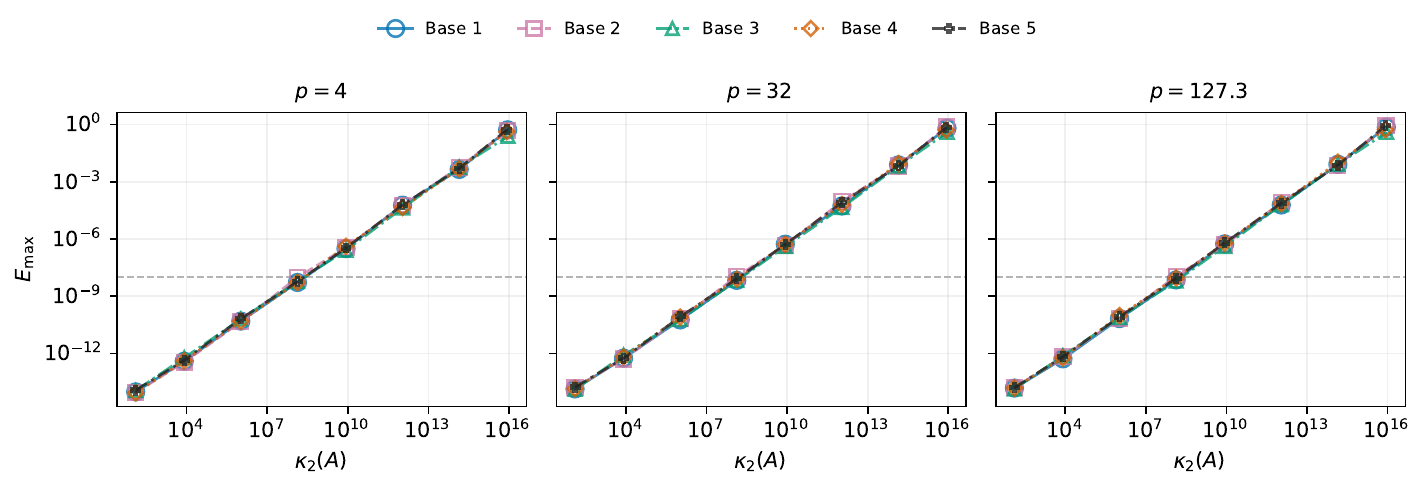}
    \caption{Finite-precision map discrepancy under increasing conditioning at representative values of $p$.}
    \label{fig:a4-map-error}
\end{figure}

As \Cref{fig:a4-map-error} shows, $E_{j,\max}(s)$ grows steadily with conditioning, reaching roughly $10^{-8}$ at $\kappa_2(\bA_{j,s})\approx10^8$. Still, the full iteration is quite stable through $s=27$ ($\kappa_2(\bA_{j,s})\approx1.3\times10^8$): all $30$ runs converge, with stopping times within about $2\%$ of their well-conditioned counterparts. At $s=33$ ($\kappa_2(\bA_{j,s})\approx8.6\times10^9$), none reaches $10^{-8}$ within $10^4$ updates. Thus, despite the growth of  $E_{j,\max}(s)$, the iteration count remains essentially unchanged over a broad conditioning range, until a finite-precision accuracy floor prevents the prescribed tolerance from being reached.

\begin{figure}[!htbp]
    \centering
    \includegraphics[width=.95\linewidth]{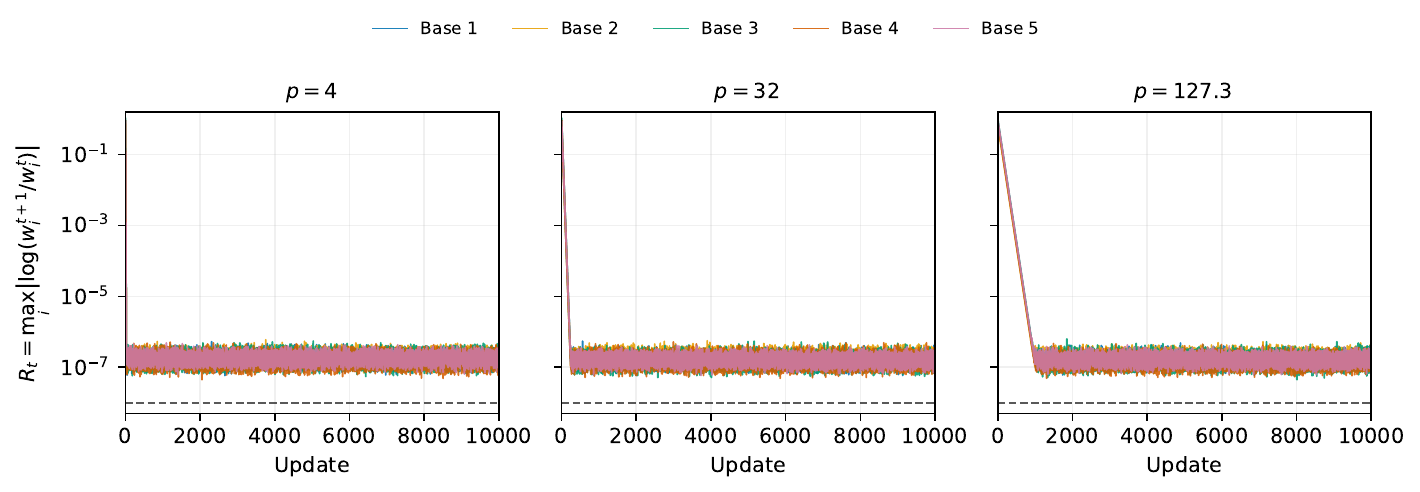}
    \caption{Residual floors under ill-conditioning at representative values of $p$.}
    \label{fig:a4-residual-floor}
\end{figure}

\Cref{fig:a4-residual-floor} explains this failure. Across all $30$ runs, the residual falls below $10^{-6}$ and then fluctuates around a $10^{-7}$-scale floor; \Cref{fig:a4-residual-floor} shows three representative values of $p$. Thus, finite precision imposes an accuracy floor above the stopping
threshold, illustrating a limitation of \Cref{alg:main}
outside the exact-arithmetic guarantee of \Cref{cor:ordinary-main}.

\subsubsection{Block Lewis weights: iteration counts}
\label{sec:expt-blkLW}
We finally ask whether the block Lewis weight iteration exhibits the same qualitative $p$-dependence as the row-wise iteration does in \Cref{fig:A1}. We generate 50 matrices $\bA\in\mathbb{R}^{100\times20}$ with i.i.d.\ standard Gaussian entries, partition each into $k=20$ blocks of five rows, and reuse the previously used values of $p$. As in \Cref{sec:A1-A2-A4}, we stop running the experiment when the maximum (among all block coordinates) absolute log-ratio between consecutive iterates falls below $10^{-8}$. 

\begin{figure}[t]
    \centering
    \includegraphics[width=.95\linewidth]{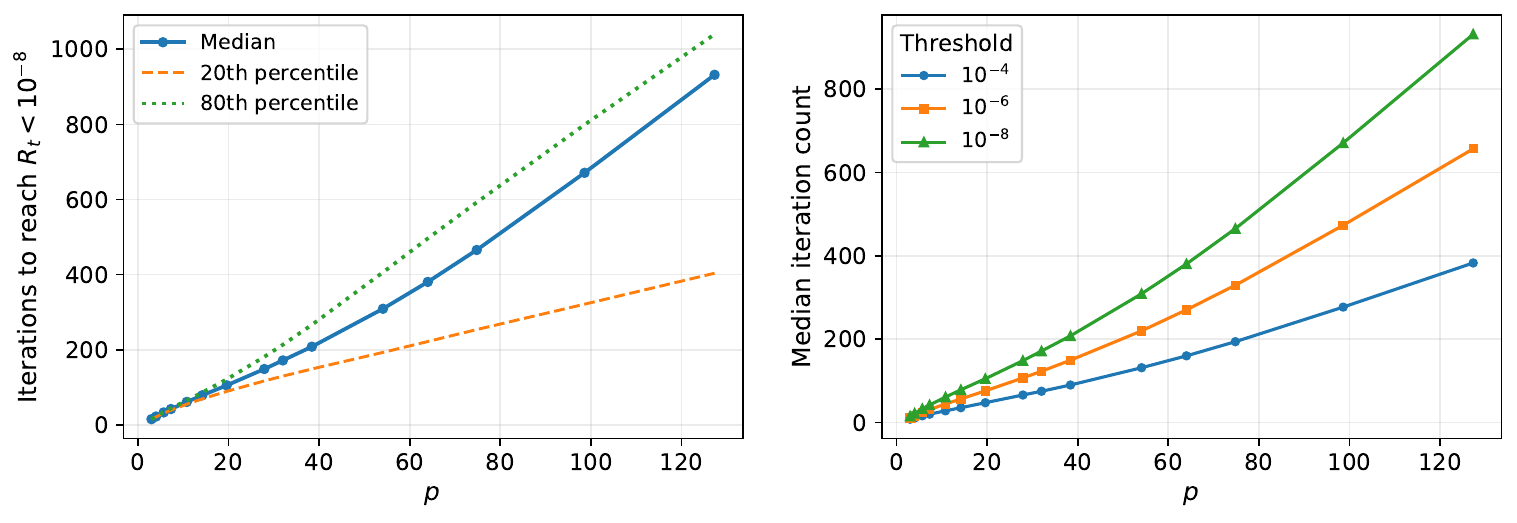}
    \caption{Block Lewis weight round scaling with $p$ (left) and sensitivity to the stopping threshold (right). }
    \label{fig:block-scaling}
\end{figure}

The left panel of \Cref{fig:block-scaling} summarizes the $10^{-8}$ stopping time across the 50 matrices by its median and empirical 20th and 80th percentiles; the right reports the median at thresholds $10^{-4}$, $10^{-6}$, and $10^{-8}$. At all three thresholds, the iteration count increases steadily with $p$, while the spread across instances widens at large $p$. Thus the tested block instances provide empirical support for the $\widetilde{O}(p)$-type dependence proved in \Cref{thm:blkLW}.

\subsection{Real-data experiments}\label{sec:real-data-expts} 
Having tested \Cref{alg:main} computationally on synthetic data, we now use our implementation to study an orthogonal property of Lewis weights, namely, how they behave when used in experiment design. We test this on the Intel Berkeley Research Lab sensor dataset~\cite{guestrin2004distributed} and meteorological data from the NOAA National Data Buoy Center.

\subsubsection{Lewis weights and experimental design}
\label{sec:expts-design-geometry}

We first study how the design allocations associated with block Lewis weights change with $p$ on two fixed sensing geometries derived from real data. In particular, we ask how much of the information gain of D-optimal design can be recovered while keeping the allocation relatively diffuse.

\looseness=-1A design is an allocation $\bpi\in\mathbb R_{\geq 0}^k$ with $\sum_{i\in[k]}\pi_i=1$, representing how a measurement budget is distributed across $k$ blocks. For each block $i\in[k]$, let $\bA_i$ denote its matrix of sensing rows and $\bC_i=\bA_i^\top\bA_i$ its information contribution. The information matrix of a design $\bpi$ is $M(\bpi):=\sum_{i\in[k]}\pi_i\bC_i$, the determinant of which is maximized in classical D-optimal design; see, e.g., \cite{pukelsheim2006optimal}.

\looseness=-1Fix a $p>2$. For the true block Lewis weights $\bwbar(p)$, define the associated design allocation $\pi_i(p):=\frac{\overline{w}_i(p)^{1-2/p}}{\sum_{j=1}^k\overline{w}_j(p)^{1-2/p}}$. These are the normalized coefficients with which the block information matrices enter the regularized D-optimal-design formulation  in \Cref{sec:apps}. The regularization favours less concentrated allocations, and the resulting family is known to interpolate from the uniform design as $p\downarrow2$ toward  block D-optimal design as $p\to\infty$. A less concentrated allocation may hedge against reliance on a few sensors; we investigate robustness to block failures separately in \Cref{sec:expts-block-failure}.

To quantify the information--concentration tradeoff, writing $L(\bpi):=\log\det M(\bpi)$, define
\[
    N_{\mathrm{eff}}(\bpi):=\frac{1}{\sum_{i\in[k]}\pi_i^2},
    \qquad
    G(\bpi):=\frac{L(\bpi)-L(U)}{L(D)-L(U)},
\]
where $U$ and $D$ denote the uniform and scalar block D-optimal designs, respectively. The effective support $N_{\mathrm{eff}}$ measures how widely the allocation is spread, with $N_{\mathrm{eff}}(U)=k$. Meanwhile, $G$ measures the fraction of the available log-determinant gain over uniform allocation, with $G(U)=0$ and $G(D)=1$. Thus, each $p$ gives a point $(N_{\mathrm{eff}}(\bpi(p)),G(\bpi(p)))$ in \Cref{fig:design-geometry}: moving left corresponds to a more concentrated allocation, while moving up corresponds to greater information gain.

\looseness=-1We evaluate this path on two real-data sensing geometries: 52 sensor blocks from the Intel Berkeley Research Lab dataset \cite{guestrin2004distributed} and 10 buoy blocks from NOAA's National Data Buoy Center. The corresponding low-rank sensing matrices have latent dimensions $r=30$ and $r=4$, respectively.
\begin{figure}[t]
    \centering
    \includegraphics[width=.95\linewidth]{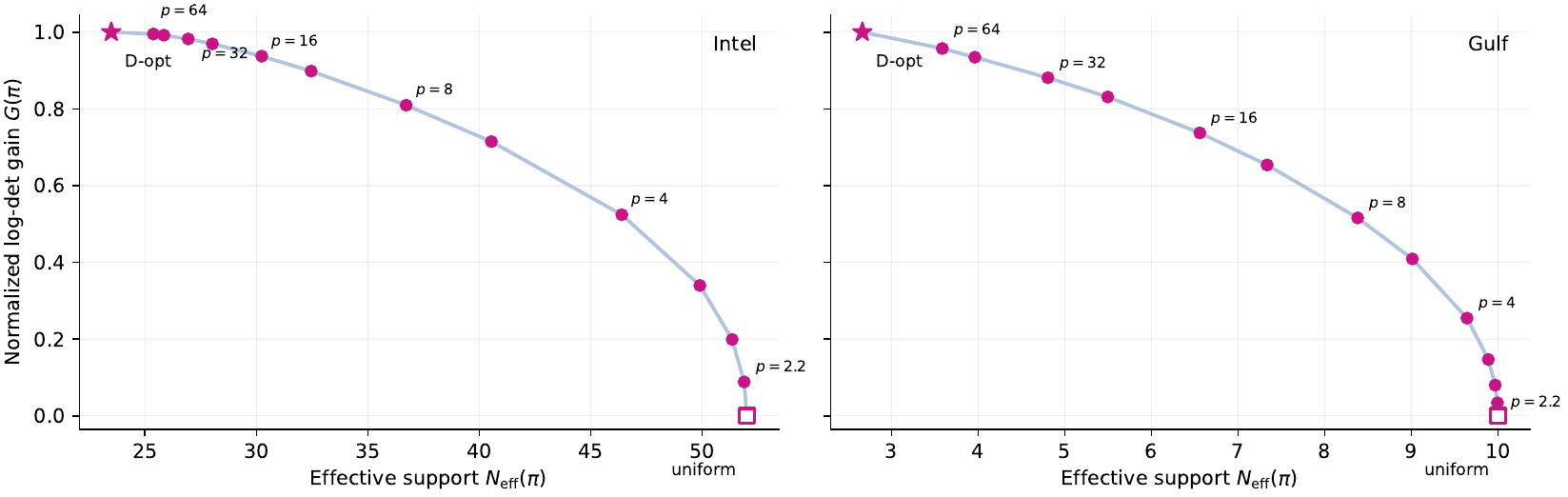}
    \caption{Finite-$p$ block Lewis design paths on Intel and Gulf; square: uniform, star: D-optimal.}
    \label{fig:design-geometry}
\end{figure}

As \Cref{fig:design-geometry} shows, increasing $p$ decreases $N_{\mathrm{eff}}(\bpi(p))$ and increases $G(\bpi(p))$ at every successive point of the evaluated $p$-grid.   Thus, much of the available log-determinant gain is recovered while the allocation remains substantially less concentrated than the D-optimal design.
\subsubsection{Information preservation under block unavailability}
\label{sec:expts-block-failure}

The interpolation in \Cref{fig:design-geometry} suggests a natural
question: does reducing concentration help preserve information when
entire blocks become unavailable? On the same Intel sensor and Gulf buoy
datasets considered above, we fix the design allocation before
availability is realized without reallocating mass from unavailable
blocks. For a surviving set $S$, write
\[
    M_S(\bpi):=\sum_{i\in S}\pi_i\bC_i,
    \qquad
    \Phi_S(\bpi):=\det(M_S(\bpi))^{1/r}.
\] \looseness=-1We vary the number $f$ of unavailable blocks, selecting these blocks uniformly at random. We report the mean
surviving information $\mu_f(\bpi):=\mathbb E[\Phi_S(\bpi)]$ and
the probability $\rho_f(\bpi)$ that $M_S(\bpi)$ loses rank.

\begin{figure}[t]
    \centering
    \includegraphics[width=.95\linewidth]{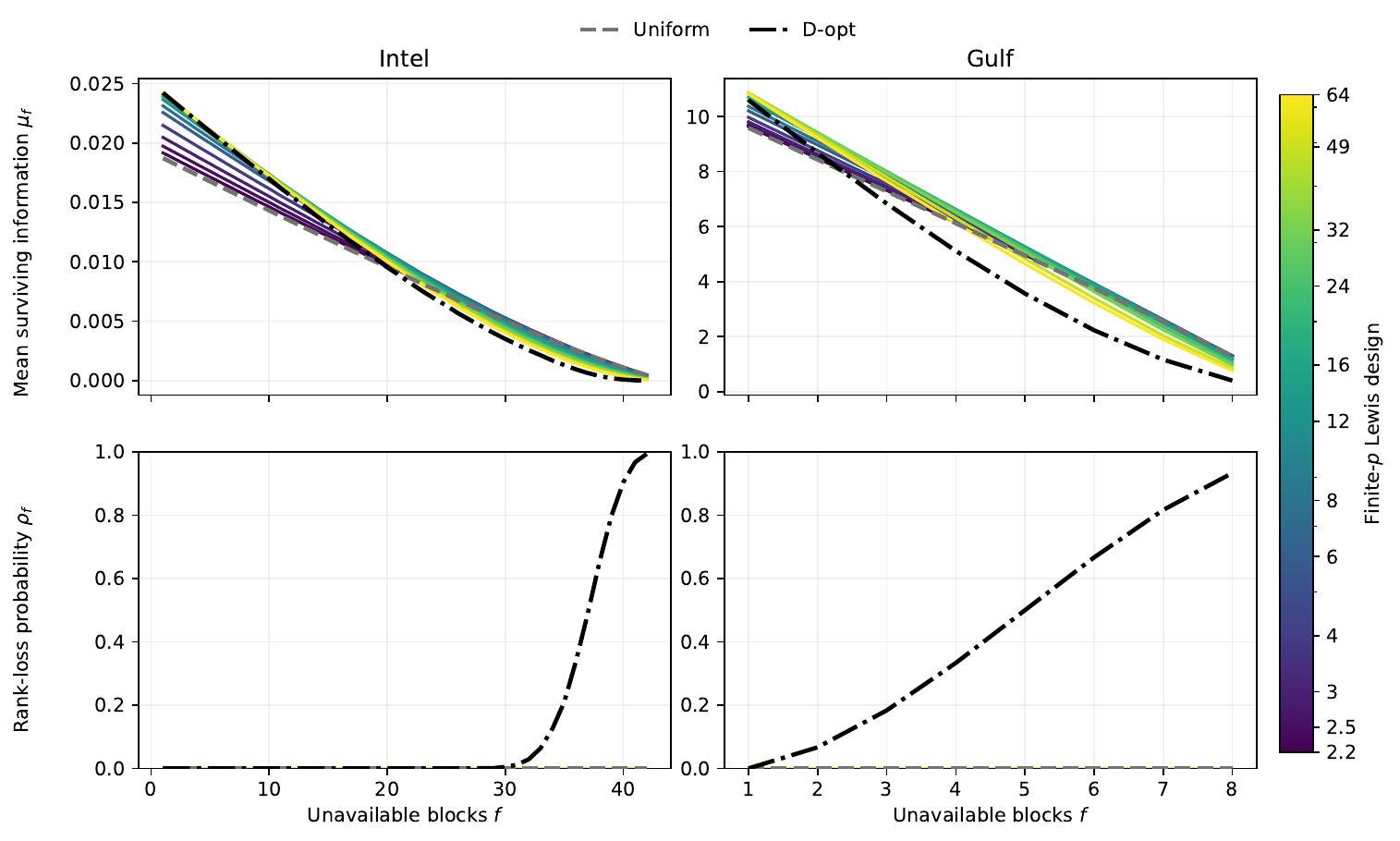}
    \caption{Mean surviving information (top) and rank-loss probability
    (bottom) under block unavailability on the Intel and Gulf datasets.}
    \label{fig:block-failure}
\end{figure}

\looseness=-1The top row of \Cref{fig:block-failure} shows the same qualitative
progression on both datasets. When few blocks are unavailable, more
concentrated, larger-$p$ designs retain the most information; as $f$
increases, the $\mu_f$ curves cross, and the maximizing design shifts
toward smaller $p$ (less concentrated allocations). 

\looseness=-1The bottom row shows that this shift is not merely due to rank loss. No rank loss was observed for uniform or any finite-$p$ design on the evaluated masks, so the crossings among their mean-information curves are not explained by rank loss. By contrast, the D-optimal endpoint begins to lose rank
as failures become severe, contributing to its deterioration. Thus, less concentrated designs can preserve more information under block unavailability, even without  rank loss.

\section*{Acknowledgements} \looseness=-1We used ChatGPT interactively in developing and verifying our proofs, code, and write-up.   We are extremely grateful to: Guanghao Ye for numerous detailed discussions about Lewis weights; Jelena Diakonikolas, Ethan Epperly, Krishnamurthy Iyer, Christopher Musco, Ankit Pensia, Arvind K. Saibaba, Kevin Tian, and Zhenyi Zhang for  generously providing detailed feedback on preliminary versions, especially for making our presentation more accessible to the randomized numerical linear algebra community; and to Li Chen for kindly coordinating arXiv updates. 

\newpage 

\printbibliography

\newpage 
\appendix 
\section{Omitted explanations of leverage scores and Lewis weights}\label{app:row-splitting}
\subsection{Leverage scores and $\ell_2$ sampling}\label{sec:appendix-levsc}

We first explain leverage scores as measures of row importance.
Given a full-column-rank matrix $\bA\in\R^{m\times n}$ with
no zero rows, consider the following question: \emph{how important
is the $i^\mathrm{th}$ row ($\ba_i^\top$) relative to the other rows in representing the
row space of $\bA$?} As explained in~\cite{cohen2015uniform}, one way to capture this notion is via leverage scores, as defined by the following optimization problem:
\[
\levsc{i}(\bA)
:=\min_{\by\in\R^m}
\left\{\|\by\|_2^2:\bA^\top\by=\ba_i\right\}, \text{ for all } i \in [m].\numberthis\label{eq:levscores-opti}
\]
\looseness=-1For example, if row $i$ is orthogonal to every other row,
the  solution to the problem described by \Cref{eq:levscores-opti} is $\by=\mathbf{e}_i$ (the all-zeroes vector with one at the $i^\mathrm{th}$ coordinate), and hence
its leverage score is one, the maximum possible. On the
other hand, if a row can be represented using several
other rows, the minimizing vector can distribute its
coefficients among them, resulting in a smaller leverage
score.

The solution of \Cref{eq:levscores-opti} is
$\by=\bA(\bA^\top\bA)^{-1}\ba_i$: indeed, every other
feasible vector differs from this one by an element of
$\ker(\bA^\top)$, which is orthogonal to the minimizer (and hence only worsens the objective).
Consequently, we obtain 
$\levsc{i}(\bA)=\ba_i^\top(\bA^\top\bA)^{-1}\ba_i$.
Equivalently, by Cauchy-Schwarz,
\[
\levsc{i}(\bA)
=\max_{\bx\neq\mathbf{0}}
\frac{|\ba_i^\top\bx|^2}{\|\bA\bx\|_2^2}.\numberthis\label{eq:levscores-var}
\]
Thus, the leverage score bounds the maximum fraction of
the squared $\ell_2$ norm contributed by a row.
Moreover, as detailed in the proof of \Cref{fact:block-levsc}, since $\bA(\bA^\top\bA)^{-1}\bA^\top$ is an
orthogonal projector of rank $n$, all leverage scores
lie in $(0,1]$ and sum to $n$.

\looseness=-1We now explain how leverage scores help construct
$\ell_2$ subspace embeddings. Imagine replacing each
row $\ba_i^\top$ with $k_i$ copies of
$\frac{\ba_i^\top}{\sqrt{k_i}}$, and denote the resulting
matrix by $\widetilde{\bA}$. Writing
$N:=\sum_{i=1}^m k_i$, we have
$\widetilde{\bA}^\top\widetilde{\bA}=\bA^\top\bA$,
so this construction preserves $\ell_2$ norms.
Moreover, each copy of row $i$ has leverage score
$\frac{\levsc{i}(\bA)}{k_i}$. Thus, choosing $k_i\propto\levsc{i}(\bA)$ equalizes the leverage
scores of all $N$ copies to
$\frac{n}{N}$.

By the variational characterization in \Cref{eq:levscores-var}, no individual
copy in $\widetilde{\bA}$ contributes more than $\frac{n}{N}$ of the squared
$\ell_2$ norm. Standard concentration arguments then show
that uniformly sampling and appropriately rescaling
$O\bigl(\frac{n\log(n/\delta)}{\varepsilon^2}\bigr)$
rows of $\widetilde{\bA}$ yields an $\ell_2$ subspace embedding with
probability at least $1-\delta$, which is optimal up to
logarithmic factors~\cite{drineas2006sampling,cohen2015uniform}.

\looseness=-1An important consequence of the above ``row-splitting'' interpretation is that uniformly sampling a copy selects original row
$i$ with probability
$\frac{k_i}{N}=\frac{\levsc{i}(\bA)}{n}$.
Thus, uniform sampling of the copies is equivalent to
leverage-score sampling of the original rows --- we next extend these ideas to Lewis weights.

\subsection{Lewis weights and $\ell_p$ sampling}\label{sec:rowsplit}

We now extend the preceding interpretation to $\ell_p$
geometry, following the explanation by Cohen and Peng~\cite{cohen2015lp}.
Fix $p>2$. As before, we replace each row $\ba_i^\top$
with $k_i$ copies, but now scale each copy by
$k_i^{-1/p}$ instead of $k_i^{-1/2}$. This again
preserves the relevant norm: for every $\bx\in\R^n$,
we have
\[\|\widetilde{\bA}\bx\|_p^p
=\sum_{i=1}^m k_i
\left(\frac{|\ba_i^\top\bx|}{k_i^{1/p}}\right)^p
=\|\bA\bx\|_p^p.\] As in the $\ell_2$ setting, we would like to equalize
the leverage scores of all copies. Indeed, if each
copy has leverage score $\frac{n}{N}$, the variational
characterization of leverage scores gives
$|\widetilde{\ba}_{(i,s)}^\top\bx|^2
\leq\frac{n}{N}\|\widetilde{\bA}\bx\|_2^2$,
where $\widetilde{\ba}_{(i,s)}^\top$ denotes the
$s^\mathrm{th}$ copy of row $i$.
Combining this with
$\|\widetilde{\bA}\bx\|_2
\leq N^{1/2-1/p}\|\widetilde{\bA}\bx\|_p$
gives
\[
\max_{\bx\neq\mathbf{0}}
\frac{|\widetilde{\ba}_{(i,s)}^\top\bx|^p}
{\|\widetilde{\bA}\bx\|_p^p}
\leq
\left(\frac{n}{N}\right)^{p/2}N^{p/2-1}
=\frac{n^{p/2}}{N}.
\]
Thus, just as for the $\ell_2$ case in \Cref{sec:appendix-levsc}, balancing leverage scores controls the maximum fraction of $\|\widetilde{\bA}\bx\|_p^p$  contributed
by any copy. Concentration arguments then show that uniformly sampling (and appropriately scaling) sufficiently many rows of $\widetilde{\bA}$ produces a matrix $\mathbf{S}\widetilde{\bA}$ satisfying $\|\mathbf{S}\widetilde{\bA}\bx\|_p\approx\|\widetilde{\bA}\bx\|_p=\|\bA\bx\|_p$ for all $\bx\in\R^n$. For fixed $p>2$ and constant accuracy, this requires $\widetilde O(n^{p/2})$ samples, which is optimal up to logarithmic factors~\cite{cohen2015lp, woodruff2025online}.

Crucially, this balancing condition is precisely what
yields the Lewis fixed-point equation. To see this,
write $\bD:=\Diag(k_1,\ldots,k_m)$. Unlike in the
$\ell_2$ case, splitting rows now changes the Gram
matrix: $\widetilde{\bA}^\top\widetilde{\bA}
=\sum_{i=1}^m k_i^{1-2/p}\ba_i\ba_i^\top
=\bA^\top\bD^{1-2/p}\bA$, in sharp contrast from what we saw in \Cref{sec:appendix-levsc}, and in fact this  is what leads to the fixed-point equation in \Cref{eq:Lewis-fixed-point}. Indeed, now the leverage score (in $\widetilde{\bA}$) of the
$s^\mathrm{th}$ copy of row $i$ of $\bA$ is
\begin{equation}
\begin{aligned}
\levsc{(i,s)}(\widetilde{\bA})
&=k_i^{-2/p}\ba_i^\top
(\bA^\top\bD^{1-2/p}\bA)^{-1}\ba_i\\
&=\frac{1}{k_i}
\levsc{i}(\bD^{1/2-1/p}\bA),
\qquad i\in[m],\ s\in[k_i].
\end{aligned}
\label{eq:row-copy-leverage}
\end{equation}
In particular, all copies of a given row have equal
leverage scores, and their sum is
$\levsc{i}(\bD^{1/2-1/p}\bA)$.

We now ask when all $N$ copies have equal leverage
scores. By \Cref{eq:row-copy-leverage}, this holds
precisely when
$\levsc{i}(\bD^{1/2-1/p}\bA)=ck_i$ for every
$i\in[m]$, for some $c>0$. Define
$w_i:=ck_i$ and $\bW:=\Diag(\bw)=c\bD$.
Since leverage scores are invariant under uniform
matrix scaling, the balancing condition becomes
$w_i=\levsc{i}(\bW^{1/2-1/p}\bA)$ for all $i\in[m]$: this is  exactly the Lewis fixed-point equation
\Cref{eq:Lewis-fixed-point}.

Conversely, choosing copy counts proportional to
the Lewis weights (when such integer counts exist)
equalizes the leverage scores of all copies.
Moreover, uniformly sampling a row of
$\widetilde{\bA}$ corresponds to sampling row $i$
of $\bA$ with probability
$\frac{k_i}{N}=\frac{\wbar_i}{n}$, i.e.,
Lewis-weight sampling. The restriction to integer
copy counts is only for exposition; the interpretation
extends to arbitrary Lewis weights.

\section{Omitted proofs of supporting lemmas}\label{sec:app-technical-lemmas}
\blocklevscfact* 
\begin{proof}
The matrix $\bB(\bB^\top \bB)^{-1}\bB^\top$ is an orthogonal projector, and therefore $\bzero\preceq \bB(\bB^\top \bB)^{-1}\bB^\top \preceq \bI$. Restricting each matrix in this inequality chain to the principal submatrix of the rows corresponding to the $i^\mathrm{th}$ block in the statement of the lemma, we have $\bzero\preceq \bB_{[i]}(\bB^\top \bB)^{-1}\bB_{[i]}^\top\preceq \bI$. Therefore, the eigenvalues of $\bB_{[i]}(\bB^\top \bB)^{-1}\bB_{[i]}^\top$ all lie in $[0, 1]$. We therefore have \[\blevsc{i}(\bB)=\Tr\left( \bB_{[i]}(\bB^\top \bB)^{-1}\bB_{[i]}^\top\right)\leq \rank\left( \bB_{[i]}(\bB^\top \bB)^{-1}\bB_{[i]}^\top\right)\leq \rank(\bB_{[i]}).\]This proves the first assertion. For the second bound, we see that \[\sum_{i=1}^k \blevsc{i}(\bB) = \sum_{i=1}^k\Tr\left( \bB_{[i]}(\bB^\top \bB)^{-1}\bB_{[i]}^\top\right) = \Tr\left( \bB(\bB^\top \bB)^{-1}\bB^\top\right) = n.\qedhere\]
\end{proof}

\begin{lemma}[Properties of $\bx\mapsto\log\det\bG(e^{\bx})$]\label{lem:logdetGexp-is-cvx} Following the notation of \Cref{sec:intro-notationPrelims}, let $\bA\in\R^{m\times n}$ be partitioned into nonzero row blocks $\bAblk{1}, \bAblk{2}, \dotsc, \bAblk{k}$, with $\bAblk{i}\in\R^{m_i \times n}$.  Given a vector $\bx\in\R^k$, define the function $f(\bx):= \log\det\bG(e^{\bx}) = \log\det\left(\bA^\top \Diag\left(e^{\omrp x_1}\bI_{m_1}, e^{\omrp x_2}\bI_{m_2}, \dotsc, e^{\omrp x_k}\bI_{m_k}\right)\bA\right)$. Then, the following properties hold. 
\begin{enumerate}
\item The function $f$ is convex. 
\item For every $i \in [k]$, we have that $\frac{\partial f(x)}{\partial x_i} = \omrp e^{\omrp x_i } \Tr\left(\bAblk{i}\bG(e^{\bx})^{-1}\bAblk{i}^\top\right) = \omrp \calT_i(e^{\bx}).$
\end{enumerate} 
\end{lemma}
\begin{proof}
When $\bA$ is a square matrix, the function $f$ is affine in $\bx$ and hence convex. For rectangular matrices, we instead invoke the Cauchy-Binet formula~\cite{horn2012matrix}. 

Given a matrix $\bB\in\R^{m\times n}$ with $m\geq n$, let $\bB_S\in\R^{n\times n}$ denote the square matrix obtained by selecting $n$ rows of $\bB$. Then the Cauchy-Binet formula says \[\det\left(\bB^\top\bB\right) = \sum_{\substack{S\subseteq[m]\\ |S| = n}} \left(\det(\bB_S)\right)^2.\numberthis\label{eq:logdetG-1}\] We proceed to apply this to $\bB : = \left(\Diag\left(e^{\omrp x_1}\bI_{m_1}, e^{\omrp x_2}\bI_{m_2}, \dotsc, e^{\omrp x_k}\bI_{m_k}\right)\right)^{1/2}\bA$. Consider a subset $S\subseteq [m]$ such that $|S|=n$. Let $z_i(S) := |S\cap \bI_i|$ denote the number of  rows of block $\bI_i$ that are present in $S$. Then, using \Cref{eq:logdetG-1} in conjunction with the fact that the determinant of a diagonal matrix is simply the product of its entries, we have $\det\bG(e^{\bx}) = \sum_{\substack{S \subseteq[m]\\ |S| = n}} \det(\bA_S)^2 e^{\omrp \cdot \sum_{i = 1}^k x_i z_i(S)}$. Consequently, we may express $f$ as \[f(\bx) = \log \left(\sum_{\substack{S \subseteq[m]\\ |S| = n}} \det(\bA_S)^2 e^{\omrp \cdot \sum_{i = 1}^k x_i z_i(S)}\right).\] We now note that, as a function of $\bx\in\R^k$, this is a log-sum-exp of affine functions of $\bx$ and  is therefore convex~\cite[Section 3.1.5]{BV04}. 
  For completeness, we now show the computation of the gradient, though this computation can be found, e.g., in \cite[Lemma 48]{lee2019solving}, \cite[Lemma 3]{fazel2022computing}. By the gradient of the determinant,  chain rule, and the cyclic permutation property of trace, we have \[\frac{\partial f(\bx)}{\partial x_i} = \Tr\left(\bG(e^{\bx})^{-1}\frac{\partial \bG(e^{\bx})}{\partial x_i}\right) = \Tr\left(\bG(e^{\bx})^{-1}\cdot \omrp e^{\omrp x_i} \bAblk{i}^\top \bAblk{i}\right) = \omrp e^{\omrp x_i}  \Tr\left(\bAblk{i} \bG(e^{\bx})^{-1} \bAblk{i}^\top\right) = \omrp \calT_i(e^{\bx}). \] 

\end{proof}

\section{Omitted proof on sensitivities}\label{sec:sens-proff}
\sensprop*
\begin{proof} The upper bound $\sens{i}(\bA)\leq n^{p/2-1}\wbar_i$ in the first part of the proposition is known \cite[Fact~3.3]{padmanabhan2023computing}. We therefore prove the reverse inequality $\wbar_i\leq\sens{i}(\bA)$.

Set $\bB\defeq\bWbar^{1/2-1/p}\bA$, and let $\bb_j^\top$ denote the $j^\mathrm{th}$ row of $\bB$. By the definition of Lewis weights, $\wbar_j=\levsc{j}(\bB)$ for every $j\in[m]$. We use the known fact that leverage scores are precisely $\ell_2$ sensitivities \cite[Section~2]{padmanabhan2023computing}, and hence
$\wbar_j=\max_{\bx\neq\bzero}\frac{|\bb_j^\top\bx|^2}{\|\bB\bx\|_2^2}$.

Fix $i\in[m]$, and let $\bx^{(i)}$ be a maximizer in the above expression for $\wbar_i$, scaled so that $\|\bB\bx^{(i)}\|_2=1$. Such a maximizer exists since $\bB$ has full column rank and the objective is homogeneous. By our choice of $\bx^{(i)}$, we have $|\bb_i^\top\bx^{(i)}|^2=\wbar_i$. Moreover, for every $j\in[m]$, since $\wbar_j$ is the maximum of the corresponding ratio and $\bx^{(i)}$ is one feasible choice with $\|\bB\bx^{(i)}\|_2=1$, we have
$|\bb_j^\top\bx^{(i)}|^2\leq\wbar_j$.

Define $z_j\defeq\frac{|\bb_j^\top\bx^{(i)}|^2}{\wbar_j}$ for every $j\in[m]$. The preceding inequalities imply $0\leq z_j\leq1$ for every $j$, while  $|\bb_i^\top\bx^{(i)}|^2=\wbar_i$ implies $z_i=1$. Further, $\sum_{j=1}^m\wbar_jz_j
=\sum_{j=1}^m|\bb_j^\top\bx^{(i)}|^2 =\|\bB\bx^{(i)}\|_2^2 =1$.

We now evaluate the $i^\mathrm{th}$ sensitivity at  $\bx^{(i)}$. Since $\bB=\bWbar^{1/2-1/p}\bA$, for every $j\in[m]$ we have $\bb_j^\top=\wbar_j^{1/2-1/p}\ba_j^\top$, or equivalently $\ba_j^\top=\wbar_j^{1/p-1/2}\bb_j^\top$. Therefore, $|\ba_j^\top\bx^{(i)}|^p = \wbar_j^{1-p/2}|\bb_j^\top\bx^{(i)}|^p$.
Using $|\bb_j^\top\bx^{(i)}|^2=\wbar_jz_j$, this becomes $|\ba_j^\top\bx^{(i)}|^p
= \wbar_j^{1-p/2}(\wbar_jz_j)^{p/2} =\wbar_jz_j^{p/2}$. Consequently,\[ \frac{|\ba_i^\top\bx^{(i)}|^p}{\|\bA\bx^{(i)}\|_p^p} =
\frac{\wbar_i z_i^{p/2}}{\sum_{j=1}^m\wbar_jz_j^{p/2}}.\]

Since $p>2$ and $z_j\in[0,1]$, this implies $z_j^{p/2}\leq z_j$ for every $j$, and hence $\sum_{j=1}^m\wbar_jz_j^{p/2}\leq\sum_{j=1}^m\wbar_jz_j=1$.
Since $z_i=1$, the preceding display therefore gives $\frac{|\ba_i^\top\bx^{(i)}|^p}{\|\bA\bx^{(i)}\|_p^p}\geq\wbar_i$. Finally, $\sens{i}(\bA)$ is the supremum of this ratio over all nonzero $\bx$, so $\sens{i}(\bA)\geq\wbar_i$, concluding the proof. 
\end{proof}

\end{document}

%% file: packages.tex
\usepackage[margin=1in]{geometry}
\usepackage{microtype}
\usepackage[T1]{fontenc}
\usepackage{lmodern}

\usepackage{amsmath,amssymb,amsthm,mathtools,xfrac}
\usepackage{bm}
\usepackage{mathrsfs}

\usepackage{enumitem}
\usepackage{algorithm}
\usepackage[noend]{algpseudocode}
\usepackage{booktabs}
\usepackage{graphicx}
\usepackage{caption}
\usepackage[dvipsnames]{xcolor}

\colorlet{myblue}{RoyalBlue!70!Blue}
\colorlet{mymagenta}{Magenta!50!Plum}

\usepackage[
    colorlinks=true,
    linkcolor=myblue,
    citecolor=mymagenta,
    urlcolor=MidnightBlue
]{hyperref}

\usepackage[
    backend=biber,
    backref=false, 
    style=alphabetic,
    sorting=ynt,        
    maxbibnames=999,   
    maxcitenames=999, 
    maxalphanames=999, 
    giveninits=false   
]{biblatex}

\DefineBibliographyStrings{english}{
  backrefpage  = {p\adddot},
  backrefpages = {p\adddot}
}

\usepackage{aliascnt}
\usepackage[nameinlink,capitalize,noabbrev]{cleveref}

%% file: formatting_macros.tex
\newcommand\numberthis{\addtocounter{equation}{1}\tag{\theequation}}

\theoremstyle{plain}

\newtheorem{theorem}{Theorem}[section]

\newaliascnt{lemma}{theorem}
\newtheorem{lemma}[lemma]{Lemma}
\aliascntresetthe{lemma}

\newaliascnt{proposition}{theorem}

\aliascntresetthe{proposition}

\newaliascnt{corollary}{theorem}
\newtheorem{corollary}[corollary]{Corollary}
\aliascntresetthe{corollary}

\newaliascnt{fact}{theorem}

\aliascntresetthe{fact}

\newaliascnt{claim}{theorem}

\aliascntresetthe{claim}

\newaliascnt{definition}{theorem}
\newtheorem{definition}[definition]{Definition}
\aliascntresetthe{definition}

\newaliascnt{assumption}{theorem}

\aliascntresetthe{assumption}

\newaliascnt{remark}{theorem}
\newtheorem{remark}[remark]{Remark}
\aliascntresetthe{remark}

\numberwithin{equation}{section}

\crefname{equation}{}{}
\Crefname{equation}{}{}
\creflabelformat{equation}{#2(#1)#3}

\crefname{theorem}{Theorem}{Theorems}
\Crefname{theorem}{Theorem}{Theorems}

\crefname{lemma}{Lemma}{Lemmas}
\Crefname{lemma}{Lemma}{Lemmas}

\crefname{proposition}{Proposition}{Propositions}
\Crefname{proposition}{Proposition}{Propositions}

\crefname{corollary}{Corollary}{Corollaries}
\Crefname{corollary}{Corollary}{Corollaries}

\crefname{fact}{Fact}{Facts}
\Crefname{fact}{Fact}{Facts}

\crefname{claim}{Claim}{Claims}
\Crefname{claim}{Claim}{Claims}

\crefname{definition}{Definition}{Definitions}
\Crefname{definition}{Definition}{Definitions}

\crefname{assumption}{Assumption}{Assumptions}
\Crefname{assumption}{Assumption}{Assumptions}

\crefname{remark}{Remark}{Remarks}
\Crefname{remark}{Remark}{Remarks}

\crefname{algorithm}{Algorithm}{Algorithms}
\Crefname{algorithm}{Algorithm}{Algorithms}
\crefname{line}{line}{lines}
\Crefname{line}{Line}{Lines}
\crefalias{ALG@line}{line}

\makeatletter
\AddToHook{env/algorithmic/begin}{\def\@currentcounter{ALG@line}}
\makeatother

%% file: macros.tex
\newcommand{\R}{\mathbb{R}}

\newcommand{\one}{\mathbf{1}}
\newcommand{\diag}{\operatorname{diag}}
\newcommand{\Diag}{\operatorname{Diag}}

\newcommand{\Tr}{\operatorname{Tr}}
\newcommand{\rank}{\operatorname{rank}}

\newcommand{\eps}{\varepsilon}
\newcommand{\defeq}{\mathrel{\vcentcolon=}}

\newcommand{\Rmspos}{\mathbb{R}^m_{>0}}
\newcommand{\Rkpos}{\mathbb{R}_{>0}^k}
\newcommand{\Prob}{\operatorname{Prob}}
\newcommand{\E}{\mathbb{E}}

\newcommand{\bA}{\mathbf{A}}
\newcommand{\bB}{\mathbf{B}}

\newcommand{\bC}{\mathbf{C}}
\newcommand{\bD}{\mathbf{D}}

\newcommand{\bG}{\mathbf{G}}

\newcommand{\bI}{\mathbf{I}}
\newcommand{\bJ}{\mathbf{J}}
\newcommand{\bK}{\mathbf{K}}

\newcommand{\bP}{\mathbf{P}}

\newcommand{\bq}{\mathbf{q}}
\newcommand{\bQ}{\mathbf{Q}}
\newcommand{\bW}{\mathbf{W}}
\newcommand{\bV}{\mathbf{V}}
\newcommand{\bX}{\mathbf{X}}
\newcommand{\bY}{\mathbf{Y}}
\newcommand{\bzero}{\mathbf{0}}

\newcommand{\ba}{\mathbf{a}}
\newcommand{\bb}{\mathbf{b}}
\newcommand{\bh}{\mathbf{h}}
\newcommand{\bw}{\mathbf{w}}
\newcommand{\bv}{\mathbf{v}}
\newcommand{\bx}{\mathbf{x}}
\newcommand{\by}{\mathbf{y}}
\newcommand{\bu}{\mathbf{u}}
\newcommand{\bxbar}{\overline{\mathbf{x}}}
\newcommand{\bwhat}{\widehat{\mathbf{w}}}

\newcommand{\bwbar}{\overline{\mathbf{w}}}

\newcommand{\bAblk}[1]{\mathbf{A}_{[#1]}}

\newcommand{\bBblk}[1]{\mathbf{B}_{[#1]}}

\newcommand{\bpi}{\bm{\pi}}

\newcommand{\what}{\widehat{w}}
\newcommand{\wbar}{\overline{w}}
\newcommand{\Wbar}{\overline{\mathbf{W}}}
\newcommand{\bWbar}{\overline{\mathbf{W}}}

\NewDocumentCommand{\sens}{g}{%
  s^{(p)}\IfValueT{#1}{_{#1}}%
}
\newcommand{\maxsens}{s_{\max}^{(p)}}

\NewDocumentCommand{\levsc}{g}{%
  \ensuremath{\tau\IfValueT{#1}{_{#1}}}%
}

\NewDocumentCommand{\blevsc}{g}{%
  \ensuremath{\tau\IfValueT{#1}{_{[#1]}}}%
}

\newcommand{\calT}{\mathcal{T}}
\newcommand{\omrp}{\gamma} 
\newcommand{\kldiv}[2]{\mathcal{D}_{\mathrm{KL}}\left(#1 \,\|\, #2\right)}
\newcommand{\KL}[2]{D_{\mathrm{KL}}\left(#1 \,\|\, #2\right)} 
\newcommand{\logdetdiv}[2]{\mathcal{D}_{-\log\det}\left(#1, #2\right)}

%% file: approxlevscores.tex
\Cref{alg:main} uses exact leverage scores.  We now show that essentially the same
algorithm works using multiplicatively approximate leverage scores as well, with an  extra normalization required to ensure that the iterates generated by the algorithm continue to have total mass $n$.

\begin{definition}[Approximate oracle]\label{def:approxLevScoresOracle} Consider the setup of \Cref{thm:blkLW}. For $\eta\geq 0$ and a given $\bw\in\Omega$, we say that $\widetilde{\bu}\in\R^k_{>0}$ is an
$\eta$-accurate block leverage-score vector if
    \[\widetilde u_i=\calT_i(\bw)e^{\zeta_i},
    \,\text{ where } |\zeta_i|\leq \eta,
    \, \text{ for all }i\in[k].\]
\end{definition}\looseness=-1 For $\eta\leq \frac{1}{2}$, this is equivalent, up to
absolute constants, to coordinatewise $(1\pm\eta)$ relative error. We do not explicitly display the algorithm corresponding to \Cref{thm:approx-levscores-blkLW} because it is a minor modification of \Cref{alg:main}, but here it is in words: First, we retain the rank preprocessing and initialization of \Cref{alg:main} --- if $R=n$, we return $\widehat{w}_i=r_i$ for all $i\in[k]$ (\Cref{line:r-equals-n-case}); otherwise, we initialize $w_i^{(0)}=\frac{nr_i}{R}$ and run
$N=\lceil p\log\left(\frac{64p\sqrt{R}}{\varepsilon}\right)\rceil$ iterations (\Cref{line:init}). At iteration $t$, given an $\eta$-accurate block leverage-score vector $\widetilde{\bu}^{(t)}$, we replace the update and final output of \Cref{alg:main} by \begin{equation}
    w_i^{(t+1)} :=   n\frac{\widetilde{u_i}^{(t)}}{\sum_{j=1}^k \widetilde{u}_j^{(t)} }, \text{ for all } i\in[k], \qquad \widehat{w}_i^{(N)}:= \left(\frac{\widetilde{u}_i^{(N-1)}}{(w_i^{(N-1)})^{1-2/p}}\right)^{p/2}.
    \label{eq:approx-Lewis-update}
\end{equation}
Since the update in \Cref{eq:approx-Lewis-update} normalizes a positive vector to have total mass $n$, and $\bw^{(0)}\in\Omega$, all iterates satisfy $\bw^{(t)} \in \Omega$. Moreover, by
\Cref{fact:block-levsc}, the update in  \Cref{eq:approx-Lewis-update} reduces to  the original Lewis update when the leverage scores are exact. We now prove \Cref{thm:approx-levscores-blkLW} using the oracle in \Cref{def:approxLevScoresOracle}. Interestingly, the proof from \Cref{sec:computingViaExacts} transfers over with only minor modifications to control the additional error term: each  update in \Cref{eq:approx-Lewis-update} is a  multiplicative perturbation of the exact update in \Cref{eq:Lewis-update}, so the KL contraction from \Cref{sec:kl-decay} persists up to an additive error. The spectral and coordinatewise recovery arguments from \Cref{{sec:coordinate-wise-control}} then apply.

\subsubsection{The approximate block Lewis update contracts KL divergence up to an error}\label{sec:proofs-approx-levsc-blkLW}
We first establish \Cref{cor:approx-kl-recurrence},  the approximate analogue of \Cref{{cor:kl-decay}}. To this end, we first establish \Cref{lem:approx-mult-perturb} via Hoeffding's lemma, H\"older's, and
Pinsker's inequalities. 
\begin{lemma}\label{lem:approx-mult-perturb} Let $\bx,\by\in\Omega$, let $|\zeta_i|\leq\eta$ for all $i\in [k]$, and define $y_i':=n\frac{y_i e^{\zeta_i}}{\sum_{j=1}^k y_j e^{\zeta_j}}$ for all $i\in [k]$. Then, the following holds: 
\[\sqrt{\kldiv{\bx}{\by'}} \leq \sqrt{\kldiv{\bx}{\by}} + \sqrt{\frac n2}\,\eta.\] 
\end{lemma}

\begin{proof}
Define a random variable $I$ distributed as
$\Pr[I=i]=\frac{y_i}{n}$ for $i\in[k]$.  Write $Z:=\zeta_I$; then, by assumption, $Z\in[-\eta,\eta]$. By the definition of $\by'$ and adding and subtracting $\sum_{i=1}^k y_i \zeta_i$ gives:
\begin{align*}
    \kldiv{\bx}{\by'} - \kldiv{\bx}{\by}  &= n\left(
        \log \mathbb{E}[e^Z]-\mathbb{E}[Z]
      \right) + \sum_{i=1}^k (y_i-x_i)\zeta_i\\ 
      &\leq  \frac{n\eta^2}{2} + \sum_{i=1}^k (y_i-x_i)\zeta_i\\
      &\leq  \frac{n\eta^2}{2}  + \eta \|\bx-\by\|_1\\ 
      &\leq  \frac{n\eta^2}{2}  + \eta \sqrt{2n \kldiv{\bx}{\by}},
\end{align*} where the first inequality uses Hoeffding lemma~\cite[Lemma 2.2]{boucheron2013nonasymptotic} (applicable since $Z$ lies in an interval of width $2\eta$), the second inequality is by  H\"older's inequality and $|\zeta_i|\leq \eta$ for all $i\in[k]$, and the final inequality is by Pinsker's inequality~\cite[Lemma 2.5]{tsybakov2009introduction} applied to $\frac{\bx}{n}$ and $\frac{\by}{n}$. Rearranging the terms of the inequality, completing the square, and taking square roots finishes the claim. 
\end{proof}

We now apply the above perturbation bound to the exact KL contraction from \Cref{lem:kl-contraction}.

\begin{corollary}
\label{cor:approx-kl-recurrence} For each $t\in \left\{0, 1, \dotsc, N-1\right\}$, suppose $\bw^{(t+1)}$ is obtained from $\bw^{(t)}$ by \Cref{eq:approx-Lewis-update}, using an $\eta$-accurate block leverage-score oracle as in \Cref{def:approxLevScoresOracle}. Then, for every such $t$, we have that $\sqrt{\kldiv{\bwbar}{\bw^{(t+1)}}} \leq \sqrt{\omrp\kldiv{\bwbar}{\bw^{(t)}}} +  \sqrt{\frac{n}{2}}\,\eta.$ This recursion implies  \[\sqrt{\kldiv{\bwbar}{\bw^{(t)}}} \leq {\sqrt{\omrp^t\kldiv{\bwbar}{\bw^{(0)}}}} + \sqrt{\frac{n}{2}}\,\eta\,p.\]
\end{corollary}

\begin{proof}
Fix $t\in\{0,1,\dotsc,N-1\}$. By our assumption, 
\Cref{def:approxLevScoresOracle}, and \Cref{eq:approx-Lewis-update},
there exist scalars $\xi_i^{(t)}$, for $i\in[k]$, with
$|\xi_i^{(t)}|\leq\eta$, such that  we have $w_i^{(t+1)} =\frac{n\calT_i(\bw^{(t)})e^{\xi_i^{(t)}}}
{\sum_{j=1}^k \calT_j(\bw^{(t)})e^{\xi_j^{(t)}}}$ for all $i\in [k]$.  Since
$\sum_{i=1}^k\calT_i(\bw^{(t)})=n$, we can now invoke \Cref{lem:approx-mult-perturb} with $\by:= \calT(\bw^{(t)})$ and $\by':=\bw^{(t+1)}$ with a perturbation of $\eta$; this, followed by \Cref{lem:kl-contraction}, yields 
\[
    \sqrt{\kldiv{\bwbar}{\bw^{(t+1)}}} \leq \sqrt{\kldiv{\bwbar}{\calT(\bw^{(t)})}} +    \sqrt{\frac{n}{2}}\eta \leq \sqrt{\omrp \kldiv{\bwbar}{\bw^{(t)}}} + \sqrt{\frac{n}{2}}\eta.
\]  Applying this inequality repeatedly  gives
$\sqrt{\kldiv{\bwbar}{\bw^{(t)}} } \leq \sqrt{\omrp^t\kldiv{\bwbar}{\bw^{(0)}}}+ \sqrt{\frac{n}{2}}\sum_{s=0}^{t-1}\sqrt{\omrp}^{t-1-s}\eta.$ We finish the claim by noting  $\sum_{s=0}^{t-1}\sqrt{\omrp}^{t-1-s}\leq\frac{1}{1-\sqrt{\omrp}} = \frac{1+\sqrt{1-2/p}}{2/p}\leq p.$ 
\end{proof}
As in \Cref{{sec:main-proof}}, to complete the proof of \Cref{thm:approx-levscores-blkLW}, we combine the KL divergence bound above with the coordinatewise recovery guarantee  in \Cref{{lem:coordinatewise-approx}}. 

\subsubsection{Proof of \Cref{thm:approx-levscores-blkLW}}
\begin{proof} The case $R=n$ is handled exactly as in \Cref{alg:main}, so assume $R>n$. Since we use the same initialization as \Cref{alg:main}, \Cref{eq:cor-con-1} applies and gives $\kldiv{\bwbar}{\bw^{(0)}} \leq n\log\left(\frac{R}{n}\right) \leq R$. Now, 
applying \Cref{cor:approx-kl-recurrence} for the first $N-1$ queries yields:
\begin{align*}
    \sqrt{\kldiv{\bwbar}{\bw^{(N-1)}}}&\leq \omrp^{(N-1)/2} \sqrt{\kldiv{\bwbar}{\bw^{(0)}}} + \sqrt{\frac{n}{2}}\eta p.  
\end{align*}  Choose $N = \lceil p\log\left(\frac{64p\sqrt{R}}{\varepsilon}\right)\rceil$  and $\eta\leq \frac{\varepsilon}{32p^2\sqrt{n}}$ and recall that $\sqrt{\omrp}\leq e^{-1/p}$. Plugging these into the preceding displayed inequality gives 
    $\sqrt{ \kldiv{\bwbar}{\bw^{(N-1)}} } < \frac{\eps}{16p}$ and  hence $ \kldiv{\bwbar}{\bw^{(N-1)}} \leq   \frac{\eps^2}{256p^2}$. Since the initialization and normalization in \Cref{eq:approx-Lewis-update} ensure  $\bw^{(N-1)}\in\Omega$, we may apply  \Cref{lem:normalized-gram-eigs-small}, in conjunction with the last inequality, to obtain \[\logdetdiv{\bG(\bw^{(N-1)})}{\bG(\bwbar)} \leq \omrp \kldiv{\bwbar}{\bw^{(N-1)}} \leq \frac{\eps^2}{256p^2}.
\] This inequality is exactly the premise of
\Cref{lem:coordinatewise-approx} with accuracy parameter $\frac{\eps}{2}$. Let $v_i := \left(\frac{\calT_i(\bw^{(N-1)})}{(w_i^{(N-1)})^\omrp}\right)^{p/2}.$ Then by \Cref{lem:coordinatewise-approx} we have 
\begin{equation}
    \left(1-\frac{\eps}{2}\right)\wbar_i \leq v_i \leq \left(1+\frac{\eps}{2}\right)\wbar_i.
    \label{eq:approx-exact-recovery}
\end{equation} 
By \Cref{def:approxLevScoresOracle}, the final oracle call satisfies $\widetilde{u}_i^{(N-1)}=\calT_i(\bw^{(N-1)})e^{\zeta_i}$ for some $|\zeta_i|\leq\eta\leq\frac{\varepsilon}{4p}$. Thus, by \Cref{eq:approx-Lewis-update} the algorithm's output  may be expressed as $\what_i = v_i e^{(p/2)\zeta_i}$, which implies $e^{-\eps/8}v_i \leq \what_i \leq e^{\eps/8}v_i$. Combining this with \Cref{eq:approx-exact-recovery} yields $e^{-\varepsilon/8}\left(1-\frac{\varepsilon}{2}\right)\wbar_i\leq \what_i \leq e^{\varepsilon/8}\left(1+\frac{\varepsilon}{2}\right)\wbar_i$. We can then prove the claim by simplifying these two inequalities: We use $e^{-x}\geq 1-x$ for $x\geq0$ to simplify the lower bound; for the upper bound, we use $e^x\leq1+2x$ for $0\leq x\leq1$. 
\end{proof}